\documentclass[runningheads,envcountsame]{llncs}
\usepackage{breakurl}

\usepackage[table]{xcolor}
\usepackage{enumerate}

\usepackage{chngcntr}
\usepackage{apptools}
\AtAppendix{\counterwithin{lemma}{section}}
\AtAppendix{\counterwithin{definition}{section}}
\AtAppendix{\counterwithin{table}{section}}
\AtAppendix{\counterwithin{equation}{section}}
\AtAppendix{\counterwithin{corollary}{section}}

\usepackage{pgf, tikz}
\usetikzlibrary{arrows, automata}

\usepackage{colonequals}
\usepackage{etoolbox}
\usepackage{amssymb}
\usepackage{amsmath}
\usepackage{paralist}
\usepackage{adjustbox}

\usepackage{textcomp}

\usepackage{mathrsfs}
\usepackage{paralist}
\usepackage{tikz}
\usepackage{tikz-cd}

\usepackage{hyperref}
\usepackage[capitalise]{cleveref}

\usepackage{float}
\usepackage{tabularx}

\usepackage[shortlabels]{enumitem}

\usepackage[clock]{ifsym}

\usepackage{stmaryrd}

\renewcommand{\bbbt}{\mathbb{T}}
\renewcommand{\bbbn}{\mathbb{N}}
\newcommand{\ce}{\colonequals}
\newcommand{\cce}{\coloncolonequals}

\newcommand{\extension}[1][]{\mathscr{E}^{#1}}
\newcommand{\pprotocols}[1][]{\envprotocols[{#1}]\times\agprotocols[{#1}]}

\newcommand{\gevents}[1][]{\mathit{GEvents}_{#1}}
\newcommand{\bevents}[1][]{\mathit{BEvents}_{#1}}
\newcommand{\fevents}[1][]{\mathit{FEvents}_{#1}}
\newcommand{\sevents}[1][]{\mathit{SysEvents}_{#1}}
\newcommand{\gtrueevents}[1][]{\overline{\mathit{GEvents}}_{#1}}
\newcommand{\gtrueactions}[1][]{\overline{\mathit{GActions}}_{#1}}
\newcommand{\gactions}[1][]{\gtrueactions[#1]}

\newcommand{\gtruethings}[1][]{\overline{\ghaps[#1]}}

\newcommand{\events}[1][]{{\mathit{Events}}_{#1}}
\newcommand{\actions}[1][]{{\mathit{Actions}}_{#1}}
\newcommand{\haps}[1][]{{\mathit{Haps}}_{#1}}
\newcommand{\ghaps}[1][]{{\mathit{GHaps}}_{#1}}
\newcommand{\gtruehaps}[1][]{\overline{\mathit{GHaps}}_{#1}}

\newcommand{\truethings}[1][]{\haps[#1]}
\newcommand{\agents}{\mathcal{A}}
\newcommand{\agentsof}[1]{\mathop{\mathcal{A}}\left(#1\right)}

\newcommand{\msgs}[1][]{\Msgs[#1]}
\newcommand{\Msgs}[1][]{\mathit{Msgs}^{#1}}

\newcommand{\glob}[2]{{\mathop{\mathit{global}}\left(\ifstrempty{#1}{\dots}{#1},\ifstrempty{#2}{\dots}{#2}\right)}}

\newcommand{\gsend}[4]{\mathop{\mathit{gsend}}(%
    \ifstrempty{#1}{\dots}{#1},%
    \ifstrempty{#2}{\dots}{#2},%
    \ifstrempty{#3}{\dots}{#3},%
    \ifstrempty{#4}{\dots}{#4}%
)}
\newcommand{\send}[2]{\mathop{\mathit{send}}(%
    \ifstrempty{#1}{\dots}{#1},%
    \ifstrempty{#2}{\dots}{#2}%
)}

\newcommand{\grecv}[4]{\mathop{\mathit{grecv}}(%
    \ifstrempty{#1}{\dots}{#1},%
    \ifstrempty{#2}{\dots}{#2},%
    \ifstrempty{#3}{\dots}{#3},%
    \ifstrempty{#4}{\dots}{#4}%
)}
\newcommand{\recv}[2]{\mathop{\mathit{recv}}(%
    \ifstrempty{#1}{\dots}{#1},%
    \ifstrempty{#2}{\dots}{#2}%
)}

\newcommand{\fakeof}[2]{{\mathop{\mathit{fake}}\left(#1,#2\right)}}
\newcommand{\localstates}[2][]{%
    {\mathscr{L}^{#1}_{#2}}%
}
\newcommand{\globalstates}{\mathscr{G}}
\newcommand{\globalinitialstates}[1][]{\mathscr{G}_{#1}(0)}

\newcommand{\failof}[1]{{\mathop{\mathit{fail}}\left(#1\right)}}
\newcommand{\sleep}[1]{{\mathop{\mathit{sleep}}\left(#1\right)}}
\newcommand{\hib}[1]{{\mathop{\mathit{hibernate}}\left(#1\right)}}

\newcommand{\tick}{\textbf{noop}}

\newcommand{\mistakefor}[2]{#1\mapsto#2}

\newcommand{\correct}[1]{correct_{#1}}
\newcommand{\faulty}[1]{faulty_{#1}}
\newcommand{\fake}[2][]{\mathit{fake}_{#1}\left(#2\right)}

\newcommand{\occurred}[2][]{\mathit{occurred}_{#1}(#2)}
\newcommand{\trueoccurred}[2][]{\overline{\mathit{occurred}}_{#1}(#2)}

\newcommand{\eventually}[1]{\lozenge{#1}}

\newcommand{\kstruct}[3]{(\Kstruct{#1},{#2},{#3})}

\newcommand{\unaware}[2]{{\mathit{unaware}(#1,#2)}}

\newcommand{\envprotocol}[1]{P_{\epsilon}\ifstrempty{#1}{}{\left(#1\right)}}
\newcommand{\agprotocol}[2]{{P_{#1}\ifstrempty{#2}{}{\left(#2\right)}}}
\newcommand{\joinprotocol}[1]{P\ifstrempty{#1}{}{\left(#1\right)}}

\newcommand{\envprotocols}[1][]{\mathscr{C}^{#1}_\epsilon}
\newcommand{\agprotocols}[1][]{\mathscr{C}^{#1}}

\newcommand{\failed}[3][]{{Failed_{#1}\left(#2\ifstrempty{#3}{}{,#3}\right)}}

\newcommand{\adversary}{adversary}

\newcommand{\alphae}[3][]{{%
    \alpha_{\epsilon_{#1}}^{#3}\ifstrempty{#2}{}{\left({#2}\right)}%
}}

\newcommand{\alphaag}[3]{{%
    \alpha_{#1}^{#3}\ifstrempty{#2}{}{\left(#2\right)}%
}}

\newcommand{\betae}[3][]{{%
    \beta_{\epsilon_{#1}}^{#3}\ifstrempty{#2}{}{\left({#2}\right)}%
}}
\newcommand{\betag}[3][]{{%
    \beta_{g_{#1}}^{#3}\ifstrempty{#2}{}{\left({#2}\right)}%
}}
\newcommand{\betab}[3][]{{%
    \beta_{b_{#1}}^{#3}\ifstrempty{#2}{}{\left({#2}\right)}%
}}
\newcommand{\betaf}[3][]{{%
    \beta_{f_{#1}}^{#3}\ifstrempty{#2}{}{\left({#2}\right)}%
}}
\newcommand{\betaout}[3][]{{%
    \overline{\beta}\strut^{#3}_{\epsilon_{#1}}\ifstrempty{#2}{}{\left({#2}\right)}%
}}

\newcommand{\betaag}[3]{{%
    \beta_{#1}^{#3}\ifstrempty{#2}{}{\left(#2\right)}%
}}

\newcommand{\filtere}[3][]{filter^{#1}_{\epsilon}\ifstrempty{#2}{}{\left(#2,#3\right)}}
\newcommand{\filterag}[4][]{filter^{#1}_{#2}\ifstrempty{#3}{}{\left(#3,#4\right)}}
\newcommand{\update}[2]{update\ifstrempty{#1}{}{\left(#1,#2\right)}}
\newcommand{\updatee}[2]{update_{\epsilon}\ifstrempty{#1}{}{\left(#1,#2\right)}}
\newcommand{\updateag}[4]{update_{#1}\ifstrempty{#2}{}{\left(#2,#3,#4\right)}}

\newcommand{\sigmaof}[1]{\sigma\ifstrempty{#1}{}{\bigl(#1\bigr)}}
\newcommand{\transition}[2][{\envprotocol{},\joinprotocol{}}]{\tau_{#1}\ifstrempty{#2}{}{\left({#2}\right)}}

\newcommand{\run}[3][]{r#1_{#2}\left(#3\right)}

\newcommand{\System}[1][]{R^{#1}}

\newcommand{\transitionExt}[3][{\envprotocol{},\joinprotocol{}}]{\tau^{#2}_{#1}\ifstrempty{#3}{}{\left({#3}\right)}}
\newcommand{\transitionfrom}[3][]{\tau^{#1}\ifstrempty{#2}{}{\left(#2\right)}\ifstrempty{#3}{}{\left(#3\right)}}

\newcommand{\tauprotocol}[3]{\tau^{#1}_{{#2},{#3}}}

\newcommand{\Admissibility}[1][b]{\Psi^{#1}}

\newcommand{\system}[1]{{R^{#1}}}
\newcommand{\EDel}{\mathit{EDel}}

\newcommand{\interpretation}[2]{\pi\ifstrempty{#1}{}{\left(#1\right)\left(#2\right)}}
\newcommand{\pwrelation}[2][]{\sim^{#1}_{#2}}

\newcommand{\Kstruct}[1]{\I^{#1}}

\newcommand{\K}[2]{K_{#1}\ifstrempty{#2}{}{#2}}

\newcommand{\everyoneK}[3][]{E^{#1}_{#2}#3}
\newcommand{\commonK}[2]{C_{#1}#2}

\newcommand{\intsys}{\I}

\newcommand{\adj}{\mathit{adj}}

\newcommand{\newfakerule}[3]{\mathit{PFake}_{#1}^{#2}\ifstrempty{#3}{}{\left({#3}\right)}}
\newcommand{\improvednewfakerule}[3]{{\mathit{BPFake}}_{#1}^{#2}\ifstrempty{#3}{}{\left({#3}\right)}}

\newcommand{\newfreezerule}[1]{\mathit{CFreeze}\ifstrempty{#1}{}{\left({#1}\right)}}

\newcommand{\newffreezerule}[2]{\mathit{BFreeze}_{#1}\ifstrempty{#2}{}{\left({#2}\right)}}

\renewcommand{\vDash}{\models}
\renewcommand{\nvDash}{\not\models}

\newcommand{\A}{\mathfrak{a}}
\newcommand{\E}{\mathfrak{e}}

\newcommand{\I}{\mathcal{I}}
\newcommand{\prop}{\mathit{Prop}}

\newcommand{\NSR}[1]{\mathit{NSR}\ifstrempty{#1}{}{\left(#1\right)}}
\newcommand{\nsr}[1]{\mathit{nsr_{#1}}}

\newcommand{\EDelC}{EDel_C}

\newcommand{\adm}{\mathbf{Adm}}
\newcommand{\jp}{\mathbf{JP}}
\newcommand{\ejp}{\mathbf{EnvJP}}
\newcommand{\efjp}{\mathbf{EvFJP}}
\newcommand{\efejp}{\mathbf{EvFEnvJP}}
\newcommand{\jpfb}{\mathbf{JP-AFB}}
\newcommand{\ejpfb}{\mathbf{EnvJP-AFB}}
\newcommand{\efjpfb}{\mathbf{EvFJP-AFB}}
\newcommand{\efejpfb}{\mathbf{EvFEnvJP-AFB}}
\newcommand{\others}{\mathbf{Others}}
\newcommand{\jpdc}{\mathbf{JP_{DC}}}
\newcommand{\ejpdc}{\mathbf{EnvJP_{DC}}}
\newcommand{\efjpdc}{\mathbf{EvFJP_{DC}}}
\newcommand{\efejpdc}{\mathbf{EvFEnvJP_{DC}}}
\newcommand{\othersdc}{\mathbf{Others_{DC}}}
\newcommand{\efejpdcmono}{\mathbf{EvFEnvJP_{DC\,mono}}}
\newcommand{\othersdcmono}{\mathbf{Others_{DC\,mono}}}
\newcommand{\ic}{\mathscr{I}}

\newcommand{\saf}{S}
\newcommand{\liv}{L}
\newcommand{\systrans}{\system{}}
\newcommand{\tracesaf}{\mathscr{T}}
\newcommand{\opsaf}{\mathscr{O}}
\newcommand{\prefset}{\mathit{PR}^{trans}}
\newcommand{\clock}{\text{\tiny\Taschenuhr}}

\newcommand{\ownsubsubsection}[1]{\medskip\noindent\textbf{#1}}

\begin{document}
\title{The Persistence of False Memory:\\ Brain in a Vat Despite Perfect Clocks%
}
%
%

\author{Thomas Schl\"ogl\orcidID{0000-0003-0037-0426} \and
Ulrich Schmid\orcidID{0000-0001-9831-8583} \and
Roman Kuznets\orcidID{0000-0001-5894-8724}}
\authorrunning{T. Schl\"ogl et al.}
%

%
\institute{TU Wien, Vienna, Austria
}
\maketitle              
\begin{abstract}
Recently, a detailed epistemic reasoning framework
for multi-agent systems with byzantine faulty asynchronous agents and possibly unreliable 
communication was introduced. We have developed a modular extension framework implemented on top of it, 
which allows to
encode and safely combine additional system assumptions commonly used in the modeling and 
analysis of fault-tolerant distributed systems, like reliable communication, time-bounded communication, multicasting, synchronous and lock-step synchronous agents and even agents with coordinated actions.
We use this extension framework for analyzing basic properties of synchronous and lock-step synchronous agents, such as the agents' local and global fault detection abilities. Moreover, we show that even
the perfectly synchronized clocks available in lock-step synchronous systems cannot be used to
avoid “brain-in-a-vat” scenarios.

\end{abstract}

\section{Introduction}
\label{sec:intro}

Epistemic reasoning is a powerful technique for modeling and analysis of distributed 
systems \cite{bookof4,HM90}, which has proved its utility also for fault-tolerant systems:
Benign faults, i.e., nodes (termed agents subsequently) that may crash and/or drop messages, have been
studied right from the beginning \cite{moses1986programming,MT88,dwork1990knowledge}. Recently,   
a comprehensive epistemic reasoning framework for agents that may even behave arbitrarily 
(``byzantine''~\cite{lamport1982byzantine}) faulty have been introduced in~\cite{KPSF19:FroCos,KPSF19:TARK}. Whereas it fully captures byzantine asynchronous systems, it is
currently not suitable for modeling and analysis of the wealth
of other distributed systems, most notably, synchronous agents
and reliable multicast communication.
 \looseness=-1

As we extend the framework~\cite{KPSF19:FroCos} in this paper, we briefly summarize its basic
notation. There is a finite set~$\agents=\{1,\dots,n\}$ (for $n \ge 2$) of \textbf{agents}, who do not have access to a global clock and execute a possibly non-deterministic joint \textbf{protocol}. In such a protocol,
agents can perform \textbf{actions}, e.g.,~send \textbf{messages} $\mu\in\msgs$, and witness \textbf{events}, in particular,~message deliveries: the action of sending a copy (numbered $k$) of a message $\mu \in \msgs$ to an agent $j\in\agents$ in a protocol is denoted by $\send{j}{\mu_k}$, whereas a receipt of such a message from  $i\in\agents$ is recorded locally as $\recv{i}{\mu}$. 
The set of all \textbf{actions} (\textbf{events}) available to an agent $i\in \agents$ is denoted by $\actions[i]$ ($\events[i]$), subsumed as \textbf{haps} $\haps[i] \ce \actions[i]\sqcup\events[i]$, with $\actions \ce \bigcup_{i\in \agents}\actions[i]$, $\events \ce  \bigcup_{i\in \agents}\events[i]$, and $\haps \ce \actions \sqcup \events$. 

The other main player in \cite{KPSF19:FroCos} is the \textbf{environment} $\epsilon$, which takes care of scheduling haps, failing agents, and resolving non-deterministic choices in the joint protocol. Since the notation above only describes the local view of agents, there is also a \textbf{global} syntactic representation of each hap, which is only available to the environment and contains additional information (regarding the time of a hap, a distinction whether a hap occurred in a correct or byzantine way, etc.) that will be detailed in Sect.~\ref{sec:model}.

The model utilizes a discrete time model, of arbitrarily fine resolution, with time 
domain $t\in \bbbt\ce\bbbn=\{0,1,\dots\}$.
All haps taking place after a \textbf{timestamp} $t\in\bbbt$ and no later than $t+1$ are grouped into a \textbf{round} denoted $t$\textonehalf{} and treated as happening simultaneously.
In order to prevent agents from inferring the global time by counting rounds, agents are generally unaware of a round, unless they perceive an event or are prompted to act by the environment.
The latter is accomplished by special system events $go(i)$, which are complemented by two more system events for faulty agents: $\sleep{i}$ and $\hib{i}$ signify a failure to activate the agent's protocol and differ in that the latter does not even wake up the agent. 
None of the \textbf{system events} $\sevents[i] \ce \{go(i), \sleep{i},\hib{i}\}$ is directly observable by agents.

Events and actions that can occur in each round, if enabled by $go(i)$, are determined by the protocols for agents and the environment, with non-deter\-min\-istic choices resolved by the \textbf{adversary} that is considered
part of the environment. A \textbf{run} $r$ is a function mapping a point in time $t$ to an $n+1$ tuple, consisting of the environment's history and local histories $r(t) = (r_{\epsilon}(t),r_1(t),\dots,r_n(t))$ representing the state of the whole system at that time $t$.
The \textbf{environment's history}~$r_\epsilon(t) \in \localstates{\epsilon}$ is a sequence of all haps that happened, in contrast to the local histories faithfully recorded in the global format.
Accordingly, $r_\epsilon(t+1) = X \colon r_\epsilon(t)$ for the set~$X \subseteq \ghaps$ of all haps from round~$t$\textonehalf{}.
The exact updating procedure including the update of the local \textbf{agent histories} is the result of a complex state transition consisting of several phases, which are described in Sect.~\ref{sec:model}.
Proving the correctness of a protocol for solving a certain distributed computing problem boils down to studying the set of runs that can be generated.\looseness=-1


In its current version, \cite{KPSF19:FroCos} only supports \emph{asynchronous} agents
and communication, where both agents and message transmission may be 
arbitrarily slow. 
Notwithstanding the importance of asynchronous distributed
systems in general~\cite{lynch1996distributed}, however, it is well-known that adding
faults to the picture renders important distributed computing problems like
consensus impossible \cite{FLP85}. There is hence a vast body of research that
relies on stronger system models that add additional assumptions. One prominent
example are \textbf{lockstep synchronous systems}, where agents take actions simultaneously 
at times $t\in\bbbn$, i.e., have
access to a perfectly synchronized global clock, and 
messages sent at time $t$ are received before time~$t+1$. It is
well-known that consensus can be solved in synchronous systems with $n \geq 3f+1$
nodes, if at most $f$ of those behave byzantine~\cite{lamport1982byzantine}.\looseness=-1

\ownsubsubsection{Related work.} 
Epistemic analysis has been successfully applied to synchronous systems
with both fault-free \cite{ben2014beyond} and benign faulty agents 
\cite{CGM14,GM18:PODC} in the past.
In~\cite{ben2014beyond}, Ben-Zvi and Moses considered the 
\emph{ordered response} problem in fault-free time-bounded distributed systems
and showed that any correct solution has to establish a certain
nested knowledge. They also introduced the syncausality relation,
which generalizes Lamport's happens-before relation \cite{Lam78}
and formalizes the knowledge gain due to ``communication-by-time''
in synchronous systems. This work was extended to tightly coordinated 
responses in \cite{BM13:ICLA,GM13:TARK}.

Synchronous distributed systems with agents suffering from benign
faults such as crashes or message send/receive omissions have already been studied
in~\cite{moses1986programming,MT88}, primarily in the
context of agreement problems \cite{dwork1990knowledge,halpern2001characterization}, 
which require some form of common knowledge. More recent results are
unbeatable consensus algorithms in synchronous systems
with crash faults \cite{CGM14}, and the discovery of the
importance of \emph{silent choirs} \cite{GM18:PODC}
for message-optimal protocols in crash-resilient systems.
By contrast, we are not aware of any attempt on the epistemic analysis of
fault-tolerant distributed systems with byzantine agents.

\ownsubsubsection{Main contributions.} In the present paper, we extend \cite{KPSF19:FroCos} by a 
modular \emph{extension framework}, which allows to encode and safely combine additional 
system assumptions typically used in the modeling and analysis of fault-tolerant 
distributed systems, like reliable communication, time-bounded communication, multicasting, 
synchronous and lock-step synchronous agents and even agents with coordinated actions.
We therefore establish the first framework that facilitates a rigorous epistemic modeling 
and analysis of general distributed systems with byzantine faulty agents.
We demonstrate its utility by analyzing some basic properties of the synchronous 
and lock-step synchronous agent extensions, namely, the agents' local and global fault 
detection abilities. Moreover, we prove that even the perfectly synchronized
clocks available in the lock-step synchronous extension cannot prevent 
a “brain-in-a-vat” scenario. 

\ownsubsubsection{Paper organization.} Additional details of the existing basic modeling framework~\cite{KPSF19:FroCos} required for our extensions are provided in Sect.~\ref{sec:model}.
Sect.~\ref{sec:synchronous_agents} presents the cornerstones of our synchronous extension and establishes some introspection results and the possibility of brain-in-a-vat scenarios.
Sect.~\ref{sec:extensions} provides an overview of our fully-fledged extension framework, the utility of which is demonstrated by investigating these issues in lock-step synchronous systems in Sect.~\ref{sec:lss_agents}.
Some conclusions in Sect.~\ref{sec:conclusions} round-off our paper.
All the material omitted from the main body of the paper due to lack of space is provided in Appendix~\ref{sec:appendix}.\footnote{%
Please note that the comprehensive appendix has been provided solely 
for the convenience of the reviewers; all material collected there 
can reasonably be omitted. It is/will of course be available in the 
existing publications and in an extended report.%
}

\section{The Basic Model}
\label{sec:model}


Since this paper extends the framework from~\cite{KPSF19:FroCos}, we first briefly 
recall the necessary details and aspects needed for defining our extension framework.
 
\ownsubsubsection{Global haps and faults.} As already mentioned in \cref{sec:intro}, there is a global version of every $\haps$ that provides
additional information that is only accessible to the environment. Among it is the timestamp $t$  of every correct action $a\in\actions[i]$, as initiated by agent $i$ in the local format, which is provided by  a one-to-one function $\glob{i,t}{a}$.
Timestamps are especially crucial for proper message processing with
$
	\glob{i,t}{\send{j}{\mu_k}} \ce \gsend{i}{j}{\mu}{id(i,j,\mu,k,t)}
$ 
for some one-to-one function $id \colon \agents \times \agents \times \msgs \times \bbbn \times \bbbt \to \bbbn$ that assigns each sent message a unique \textbf{global message identifier} (GMI). 
These GMIs enable the direct linking of send actions to their corresponding delivery events, most importantly used to ensure that only sent messages can be delivered (causality).
The resulting sets $\gtrueactions[i] \ce \{\glob{i,t}{a} \mid t\in \bbbt, a \in \actions[i]\}$ of correct actions in global format are pairwise disjoint due to the injectivity of $\mathit{global}$.
We set $\gtrueactions \ce \bigsqcup_{i\in\agents}\gtrueactions[i]$.

Unlike correct actions, correct events witnessed by agent $i$ are generated by the environment $\epsilon$, hence are already produced in the global format $\gtrueevents[i]$. 
We define $\gtrueevents \ce \bigsqcup_{i\in \agents} \gtrueevents[i]$ assuming them to be pairwise disjoint and $\gtruethings = \gtrueevents \sqcup \gtrueactions$.
A byzantine event is an event that was perceived by an agent despite not taking place.
In other words, for each correct event $E \in \gtrueevents[i]$, we use a faulty counterpart $\fakeof{i}{E}$ and will make sure that agent $i$ cannot distinguish between the two. 
An important type of correct global events is delivery $\grecv{j}{i}{\mu}{id}\in \gtrueevents[i]$ of message $\mu$ with GMI $id \in \bbbn$ sent from agent $i$ to agent $j$. 
The GMI must be a part of the global format (especially for ensuring causality) but cannot be part of the local format because it contains information about the time of sending, which should not be accessible to agents.
The stripping of this information before updating local histories is achieved by the function 
$
	\mathit{local} \colon  \gtruethings \longrightarrow  \haps
$
converting \textbf{correct} haps from the global into the local formats for the respective agents in such a way that $\mathit{local}$ reverses $\mathit{global}$, i.e., $\mathit{local}\bigl(\glob{i,t}{a}\bigr) \ce a$, in particular, $\mathit{local}{\bigl(\grecv{i}{j}{\mu}{id}\bigr)} \ce \recv{j}{\mu}$. 

To allow for the most flexibility regarding who is to blame for an erroneous action, faulty actions are modeled as byzantine events of the form $\fakeof{i}{\mistakefor{A}{A'}}$ where $A, A' \in \gtrueactions[i] \sqcup\{\tick\}$ for a special \textbf{non-action} $\tick$ in global format. 
These byzantine events are controlled by the environment and correspond to an agent violating its protocol by performing the action $A$ (in global format), while recording in its local history that it either performs $a' = \mathit{local}(A')\in \actions[i]$ if $A' \in \gactions[i]$ or does nothing if $A' = \tick$ (note that performing $A = \tick$ means not acting). 
The byzantine inaction $\failof{i}$ defined as $\fakeof{i}{\mistakefor{\tick}{\tick}}$ can be used to make agent $i$ faulty without performing any actions and without leaving a record in $i$'s local history. 
The set of all $i$'s~byzantine events, corresponding to both faulty events and actions, is denoted by $\bevents[i]$, with $\bevents \ce \bigsqcup_{i\in\agents} \bevents[i]$.
To summarize, $\gevents[i] \ce \gtrueevents[i] \sqcup \bevents[i] \sqcup \sevents[i]$ with 
$\gevents\ce \bigsqcup_{i\in\agents} \gevents[i]$, $\ghaps \ce \gevents \sqcup \gactions$.
Horizontal bars signify phenomena that are correct, as contrasted by those that may be correct or byzantine.

\ownsubsubsection{Protocols, state transitions and runs.} The events and actions that occur in each round are determined by protocols (for agents and the environment) and non-determinism (adversary).
Agent $i$'s \textbf{protocol} 
$\agprotocol{i}{} \colon \localstates{i} \to 2^{2^{\actions[i]}}\setminus\{\varnothing\}$ 
provides a range $\agprotocol{i}{r_i(t)}$ of sets of actions based on $i$'s current local state~$r_i(t) \in \localstates{i}$ at time $t$ in run $r$, from which the adversary non-deterministically picks one.
Similarly the environment provides a range of (correct, byzantine, and system) events via its protocol 
$\envprotocol{} \colon \bbbt \to 2^{2^{\gevents}}\setminus \{\varnothing\}$, 
which depends on a timestamp~$t\in\bbbt$ but \textbf{not} on the current state, in order to maintain its impartiality.
%
It is required that all events of round~$t$\textonehalf{} be mutually compatible at time~$t$, called $t$-coherent (for details see Appendix, Def.~\ref{def:t-coherence}). The set of all global states is denoted by~$\globalstates$.\looseness=-1

Agent $i$'s local view of the system after round $t$\textonehalf{} is recorded in $i$'s \textbf{local state}~$r_i(t+1)$, also called $i$'s \textbf{local history}, sometimes denoted $h_i$, which is 
agent~$i$'s share of the global state $h=r(t) \in \globalstates$.
$r_i(0)\in\Sigma_i$ are the \textbf{initial local states}, with $\globalinitialstates \ce \prod_{i\in \agents}\Sigma_i$. 
If a round contains neither  $go(i)$ nor any event to be recorded in  $i$'s local history, then the history $r_i(t+1)=r_i(t)$ remains unchanged, denying the agent knowledge that the round just passed. 
Otherwise, $r_i(t+1) = X \colon r_i(t)$, for the set $X \subseteq \haps[i]$ of all actions and events perceived by~$i$ in round~$t$\textonehalf{}, where $\colon$~stands for concatenation.  
Thus the local history~$r_i(t)$ is a sequence of all haps as perceived by~$i$ in rounds it was \textbf{active} in.

Given the \textbf{joint protocol}~$P\ce(P_1,\dots,P_n)$ and the environment's protocol~$P_\epsilon$, we focus on \textbf{$\transitionExt{}{}$-transitional runs}~$r$ that result from following these protocols and are built according to a \textbf{transition relation}~$\transitionExt{}{}\subseteq \globalstates \times \globalstates$.
Each such transitional run begins in some initial global state $r(0)\in\globalinitialstates$ and progresses, satisfying $(\run{}{t},\run{}{t+1}) \in \transitionExt{}{}$ for each timestamp $t\in\bbbt$. 
The transition relation $\transitionExt{}{}$ consisting of  five consecutive  phases is graphically represented in Appendix, Fig.~\ref{fig:trans_rel}  and described in detail below:

In the \textbf{protocol phase} a range 
$\envprotocol{t} \subset 2^{\gevents}$ 
of $t$-coherent sets of events is determined by the environment's protocol $P_\epsilon$.
Similarly for each $i\in\agents$, a range 
$\agprotocol{i}{\run{i}{t}}\subseteq 2^{\actions[i]}$ 
of sets of $i$'s actions is determined by the joint protocol~$P$.\looseness=-1

In the \textbf{adversary phase}, the adversary non-deterministically chooses a set $X_\epsilon \in \envprotocol{t}$ and one set $X_i \in \agprotocol{i}{\run{i}{t}}$ for each $i\in\agents$.

In the \textbf{labeling phase}, actions in the sets $X_i$ are translated into the global format: $\alphaag{i}{r}{t}\ce \{\glob{i,t}{a} \mid a \in X_i\}\subseteq \gtrueactions[i]$.

In the \textbf{filtering phase}, filter functions remove all unwanted or impossible attempted events from $\alphae{r}{t}\ce X_\epsilon$ and actions from $\alphaag{i}{r}{t}$. 
This is done in two stages:

\noindent
First, $\filtere{}{}$ filters out ``illegal'' events.
This filter will vary depending on the concrete system assumptions (in the byzantine asynchronous case, ``illegal'' constitutes receive events that violate causality).
The resulting set of events to actually occur in round $t$\textonehalf{} is $\betae{r}{t} \ce \filtere{\run{}{t}}{\alphae{r}{t}, \alphaag{1}{r}{t}, \dots, \alphaag{n}{r}{t}}$.
\begin{definition} \label{def:std_ac_filter}
	The \textbf{standard action filter} $\filterag[B]{i}{X_1, \dots, X_n}{X_\epsilon}$ for $i \in \agents$ either removes all actions from $X_i$ when $go(i) \notin X_\epsilon$ or else leaves $X_i$ unchanged.\looseness=-1
\end{definition}
\noindent
Second, $\filterag[B]{i}{}{}$ for each $i$ returns the sets of actions to be actually performed by agents in round $t$\textonehalf{}, i.e., $\betaag{i}{r}{t} \ce \filterag[B]{i}{\alphaag{1}{r}{t}, \dots, \alphaag{n}{r}{t} }{\betae{r}{t}}$.
Note that $\betaag{i}{r}{t} \subseteq \alphaag{i}{r}{t} \subseteq \gtrueactions[i]$ and $\betae{r}{t} \subseteq \alphae{r}{t} \subset \gevents$.

In the \textbf{updating phase}, the events $\betae{r}{t}$ and actions $\betaag{i}{r}{t}$ are appended to the global history $r(t)$.
For each $i \in \agents$, all non-system events from $\betae[i]{r}{t} \ce \betae{r}{t}\cap\gevents[i]$ and all actions $\betaag{i}{r}{t}$ as \textbf{perceived} by the agent are appended in the local form to the local history $r_i(t)$.
Note the local history may remain unchanged if no events trigger an update
(see Appendix, Def.~\ref{def:state-update} for more details).\looseness=-1
\begin{align} \label{eq:run_trans_env}
	\run{\epsilon}{t+1} &\ce \updatee{\run{\epsilon}{t}}{\quad\betae{r}{t},\quad \betaag{1}{r}{t},\quad \dots,\quad \betaag{n}{r}{t}} \\
	\run{i}{t+1} &\ce \updateag{i}{\run{i}{t}}{\quad\betaag{i}{r}{t}}{\quad\betae{r}{t}}. \label{eq:run_trans_ag}
\end{align}

The operations in the phases 2--5 (adversary, labeling, filtering and updating phase) are grouped into a \textbf{transition template} $\tau$ that yields a transition relation $\transitionExt{}{}$ for any joint and environment protocol $P$ and $P_\epsilon$.
Particularly, we denote as~$\tau^B$ the transition template utilizing $\filtere[B]{}{}$ and $\filterag[B]{i}{}{}$ (for all $i \in \agents$).\looseness=-1

As \textbf{liveness properties} cannot be ensured on a round-by-round basis, they are enforced by restricting the allowable set of runs via \textbf{admissibility conditions} $\Psi$, which are subsets of the set $\system{}$ of all \textbf{transitional runs}. 

A \textbf{context} $\gamma=(\envprotocol{},\globalinitialstates,\tau,\Psi)$ consists of an environment's protocol $\envprotocol{}$, a set of global initial states $\globalinitialstates$, a transition template $\tau$, and an admissibility condition $\Psi$.
For a joint protocol $\joinprotocol{}$, we call $\chi=(\gamma,\joinprotocol{})$ an \textbf{agent-context}.
A run $r \in \system{}$ is called \textbf{weakly $\chi$-consistent} if $r(0) \in \globalinitialstates$ and the run is $\transitionExt{}{}$-transitional.
A weakly $\chi$-consistent run $r$ is called \textbf{(strongly)} $\chi$-\textbf{consistent} if $r \in \Psi$.
The set of all $\chi$-consistent runs is denoted $\system{\chi}$.
An agent-context $\chi$ is called \textbf{non-excluding} if any finite prefix of a weakly $\chi$-consistent run can be extended to a $\chi$-consistent run.
(For more details see Appendix, Defs.~\ref{def:consistency}--\ref{def:nonexcluding}.)\looseness=-1

\ownsubsubsection{Epistemics.}\label{page:epistemics} \cite{KPSF19:FroCos} defines interpreted systems in this framework as Kripke models for multi-agent distributed environments \cite{bookof4}. 
The states in such a Kripke model are given by global histories $r(t')\in \globalstates$ for runs $r \in \system{\chi}$ given some agent-context $\chi$ and timestamps $t' \in \bbbt$.
A \textbf{valuation function} 
$\pi \colon \prop \to 2^{\globalstates}$ 
determines states where an atomic proposition from $\prop$ is true.
This determination is arbitrary except for a small set of \textbf{designated atomic propositions}:
For $\fevents[i]\ce \bevents[i]\sqcup\{\sleep{i},\hib{i}\}$, $i\in\agents$, $o \in \haps[i]$, and $t\in \bbbt$ such that $t \leq t'$,
\begin{compactitem} 
\label{page:occurred}
	\item $\correct{(i,t)}$ is true at $r(t')$ if{f} no faulty event happened to $i$ by timestamp $t$, i.e., no event from $\fevents[i]$ appears in $r_\epsilon(t)$,
	\item $\correct{i}$ is true at $r(t')$ if{f} no faulty event happened to $i$ yet, i.e., no event from $\fevents[i]$ appears in $r_\epsilon(t')$,
	\item $\fake[(i,t)]{o}$ is true at $r(t')$ if{f} $i$ has a \textbf{faulty} reason to believe  that $o\in\haps[i]$ occurred in round  $(t-1)$\textonehalf{}, i.e., $o \in r_i(t)$ because (at least in part) of some $O \in \bevents[i] \cap \betae[i]{r}{t-1}$,
	\item $\trueoccurred[(i,t)]{o}$ is true at $r(t')$ if{f} $i$ has a \textbf{correct} reason to believe $o\in\haps[i]$ occurred in round  $(t-1)$\textonehalf{}, i.e., $o \in r_i(t)$ because (at least in part) of $O \in (\gtrueevents[i] \cap \betae[i]{r}{t-1}) \sqcup  \betaag{i}{r}{t-1}$,
	\item $\trueoccurred[i]{o}$ is true at $r(t')$ if{f} at least one of $\trueoccurred[(i,m)]{o}$ for $1 \leq m \leq t'$ is; also $\trueoccurred{o}\ce \bigvee_{i\in\agents} \trueoccurred[i]{o}$,
	\item $\occurred[i]{o}$ is true at $r(t')$ if{f} either $\trueoccurred[i]{o}$ is or at least one of $\fake[(i,m)]{o}$  for $1 \leq m \leq t'$ is.
\end{compactitem}

The following terms are used to categorize agent faults caused by the environment's 
protocol~$\envprotocol{}$: agent~$i \in \agents$ is
\emph{fallible} if for any $X \in \envprotocol{t}$, $X \cup  \{\failof{i}\}  \in \envprotocol{t}$;
\emph{delayable} if     $X \in \envprotocol{t}$ implies $X \setminus \gevents[i]  \in \envprotocol{t}$;
\emph{gullible} if $X \in \envprotocol{t} $ implies that, for any $Y \subseteq  \fevents[i]$,  the set
   $Y \sqcup (X \setminus \gevents[i])  \in \envprotocol{t}$ whenever it is $t$-coherent.
Informally, fallible agents can be branded byzantine at any time;
delayable agents can always be forced to skip a round completely (which does not make them faulty);  
gullible agents can exhibit any faults in place of correct events. 
Common types of faults, e.g.,~crash or omission failures, can be obtained by restricting allowable sets~$Y$ in the definition of gullible agents.

An \textbf{interpreted system} is a pair $\I = (\system{\chi}, \pi)$.
The following BNF defines the \textbf{epistemic language} considered throughout this paper, for $p \in \prop$ and $i\in\agents$: $\varphi \cce p \mid \lnot \varphi \mid (\varphi \land \varphi) \mid K_i \varphi$ (other Boolean connectives are defined as usual; belief $B_{i} \varphi \ce K_{i} (\correct{i} \rightarrow \varphi)$ and hope $H_{i} \varphi \ce \correct{i} \to B_i \varphi$.
The interpreted systems semantics  is defined as usual with global states  $r(t)$ and $r'(t')$ indistinguishable for  $i$ if{f} $r_i(t)=r'_i(t')$ (see Appendix, Defs.~\ref{def:possible_world_relation}--\ref{def:semantics}  for the exact details).\looseness=-1


Unless stated otherwise, the global history $h$ ranges over $\globalstates$, and $X_\epsilon \subseteq \gevents$ and $X_i \subseteq \gactions[i]$ for each $i \in \agents$. The tuple $X_1,\dots,X_n$ is abbreviated $X_\agents$. We use $X_{[i,j]}$ for the tuple $X_i,\dots,X_j$ and $X'_\agents$ for $X'_1,\dots,X'_n$. For instance, $X_\agents = X_{[1,n]}$. Further, $\envprotocols$ and $\agprotocols$ are sets of all environment and joint protocols respectively.

\section{Synchronous Agents} \label{sec:synchronous_agents}

\emph{Synchronous agents}, i.e., agents who have access to a global clock that can be used to synchronize actions,
is a common type of distributed systems. All correct agents perform their actions at the same time here, with the time between consecutive actions left arbitrary and bearing no relation with message delays 
(except for lock-step synchronous agents, see Sect.~\ref{sec:lss_agents}).
A natural malfunction for such an agent is losing synch with the global clock, however,
so \emph{byzantine synchronous agents} can err by both lagging behind and running ahead of the global clock.
We implement this feature by means of \textbf{synced rounds}: correct agents act in a round $t$\textonehalf{} if{f} the round is synced, whereas a faulty agent may skip a synced round and/or act in between synced rounds. Note, however, that 
the agents do not \emph{a priori} know whether any given (past, current, or future) round is synced.
 \looseness=-1
\begin{definition}
	\label{def:virtRound}
	A round $t$\textonehalf{} is a \textbf{synced round} of a run $r \in \System$  if{f} $ \betag[i]{r}{t} \ne \varnothing$, where $\betag[i]{r}{t} = \betae[i]{r}{t} \cap \sevents[i]$, for all $i \in \agents$.
	We denote the number of synced rounds in $h \in \globalstates$ by $\NSR{h}$.
\end{definition}
In other words, a synced round requires from each agent~$i$  either $go(i)$ or one of two sync errors $\sleep{i}$ or $\hib{i}$.
Conversely, the permission $go(i)$ to act correctly should only be given during synced rounds, which we implement via the following event filter function: 
%
%
%
\begin{definition}
\label{def:filterSyn}
The \textbf{event filter function $\filtere[S]{h}{X_\epsilon, X_\agents}$ for the synchronous agents extension} outputs $X_\epsilon \setminus \{go(i) \mid i \in \agents\}$ if $\sevents[j] \cap X_\epsilon = \varnothing$ for some $j \in \agents$, or else leaves  $X_\epsilon$ unchanged.
\end{definition}
Since it is important for correct agents to be aware of synced rounds, we require agent protocols to issue the special internal action $\clock$ whenever activated:


\begin{definition}
\label{def:synch_ag_protocols}
 The set of \textbf{synchronous joint protocols} is
    \begin{align*}
        \agprotocols[S] \ce
		\left\{
			(\agprotocol{1}{},\dots \agprotocol{n}{}) \in \agprotocols \mid
			(\forall i \in \agents)(\forall h_i \in \localstates{i})(\forall D \in \agprotocol{i}{h_i})\  \clock \in D
		\right\}.
    \end{align*}
\end{definition}
The action $\clock$ enables correct agents to distinguish between an active round requiring no actions and a passive round with no possibility to act.\footnote{For the formal statement of this distinction, see Appendix, Lemma \ref{lem:synch_gbyz_indist}.}
The choices  behind our implementation of synchronicity will become clearer in Sect.~\ref{sec:extensions}.

\ownsubsubsection{Run modification} \cite{KPSF19:FroCos} is a crucial technique for proving  agents' ignorance of a fact, by creating an indistinguishable  run falsifying this fact. First, we define what it means for an agent to become byzantine. \looseness=-1
%
%
Given a run~$r$ and timestamp~$t$, a node $(i,t') \in \agents \times \bbbt$ belongs to 
the set $\failed{r}{t}$ of \textbf{byzantine nodes}, i.e., agent~$i$ is byzantine in $r$ by time $t'$, if{f} the global history $r(t')$ contains at least one  event from $\fevents[i]$.
%

\begin{definition}[Run modifications] \label{def:run_mod}
A function 
$\rho \colon  \system{\chi} \mapsto 2^{\gtrueactions[i]} \times 2^{\gevents[i]}$ 
is called an \textbf{$i$-inter\-ven\-tion for an agent-context $\chi$ and agent $i\in\agents$}. 
A \textbf{joint intervention} $B= (\rho_1,\dots,\rho_n)$ consists of $i$-interventions $\rho_i$ for each agent $i \in \agents$. 
An \textbf{adjustment}   $[B_t;\dots;B_0]$  (with \textbf{extent} $t$) is a sequence of joint interventions
$B_0,\dots,B_t$ 
to be performed at rounds $0$\textonehalf{}, \dots, $t$\textonehalf{} respectively.
\end{definition}
An  $i$-intervention $\rho(r)=(X, X_{\epsilon})$ applied to a round $t$\textonehalf{} of a given run~$r$ is intended to modify the results of this round for~$i$ in such a way that  $\betaag{i}{r'}{t}=X$ and $\betae[i]{r'}{t}=X_{\epsilon}$ in the artificially constructed new run~$r'$. We denote $\A\rho(r) \ce X$ and $\E\rho(r) \ce X_{\epsilon}$. Accordingly, a joint intervention $(\rho_1,\dots,\rho_n)$ prescribes actions $\betaag{i}{r'}{t}=\A\rho_i(r)$  for each agent~$i$ and events $\betae{r'}{t}=\bigsqcup_{i\in\agents} \E\rho_i(r)$ for the round in question. Thus, an adjustment $[B_t;\dots;B_0]$ fully determines actions and events in the initial $t+1$~rounds of the modified run $r'$:\looseness=-1

\begin{definition}
    \label{def:build_from}
    Let
    $\label{eq:generic_adj}
    \adj = \left[
    B_t;\dots;B_0
    \right]$  be an adjustment  with     $
    B_{m} = (\rho_1^{m}, \dots, \rho_n^{m})$
    for each $0 \leq m \leq t$ 
    and each $\rho_i^{m}$ be an $i$-intervention for an agent-context $\chi=\left((\envprotocol{},\globalinitialstates,\tau,\Psi),P\right)$.
The set $R({\transitionExt{}{}},{r},{\adj})$ consists of all runs~$r'$ \textbf{obtained from}  $r \in \system{\chi}$ \textbf{by adjustment}~$\adj$, i.e., runs $r'$ such that 
\begin{compactenum}[(a)]
\item           
	$\run[']{}{0} = \run{}{0}$,  
\item\label{clause:run_adj_ag}               
	$
            \run[']{i}{t'+1} = \mathit{update}_i\bigl(\run[']{i}{t'},\,\A\rho_i^{t'}(r),\,\bigsqcup_{i \in\agents}\E\rho_i^{t'}(r)\bigr)
        $ for all $i \in \agents$ and $t'\leq t$,
\item                
	$
            \run[']{\epsilon}{t'+1} =
            \mathit{update}_\epsilon\bigl(\run[']{\epsilon}{t'},\,\bigsqcup_{i \in\agents}\E\rho_i^{t'}(r),\, \A\rho_1^{t'}(r),\, \dots,\,\A\rho_n^{t'}(r)\bigr)$ for all $t'\!\leq\! t$,
\item                     
	$
            r'(T')\,\,\transitionExt{}{}\,\,r'(T'+1)$ for all $T' > t$.
\end{compactenum}
\end{definition}
The main interventions we use are as follows, where $\betab[i]{r}{t} = \betae[i]{r}{t} \cap \bevents[i]$:
\begin{definition} \label{def:bitv_interv}
    For $i\in\agents$ and  $r\in \system{}$, the interventions 
 $
        \newfreezerule{r} \ce (\varnothing, \varnothing)
        $  and
   $
        \newffreezerule{i}{r}\ce (\varnothing, \{\failof{i}\})
        $
         freeze agent~$i$ with and without fault respectively.\looseness=-1
    \begin{equation*} \label{eq:fakerule}
		\begin{gathered}
			\newfakerule{i}{t}{r}\qquad\ce\qquad \Bigl(\varnothing, \quad\betab[i]{r}{t} \ \cup \ \left\{\fakeof{i}{E} \mid E \in \betaout[i]{r}{t} \right\} \quad \cup \\
			\left\{\fakeof{i}{\mistakefor{\tick}{A}} \mid A\in \betaag{i}{r}{t}  \right\} \sqcup \left\{ \sleep{i}  \mid \betag[i]{r}{t} \in \{\{go(i)\}, \{\sleep{i}\}\}\right\} \sqcup \\
			  \{ \hib{i} \mid \betag[i]{r}{t} \notin \{\{go(i)\}, \{\sleep{i}\}\}\}\Bigr)
		\end{gathered}
    \end{equation*}
    turns all correct actions and events into indistinguishable byzantine events.
\end{definition}

Until the end of this section, we assume that $\envprotocol{}\in\envprotocols$ and  $\joinprotocol{}^S \in \agprotocols[S]$ are protocols for the environment and synchronous agents, that   $\chi = ((\envprotocol{}, \globalstates(0), \transitionfrom[S]{}{}, \system{}), \joinprotocol{}^S)$ is an agent-context where $\transitionfrom[S]{}{}$ uses the synchronous event filter from Def.~\ref{def:filterSyn} and the standard action filters from Def.~\ref{def:std_ac_filter}, and that $\I=(R^\chi,\pi)$ is an interpreted system. We additionally assume  that $P_\epsilon$ makes a fixed agent $i$, called the ``brain,'' gullible and all other agents $j\ne i$ delayable and fallible.

\begin{lemma}[Synchronous Brain-in-the-Vat Lemma]
\label{lem:synch_compose-fake-freeze}
	Consider the adjustment $\adj=[B_{t-1}; \dots ;B_0]$ with	$B_m = (\rho^m_1,\, \dots,\, \rho^m_n)$
	where 
%
%
    $
    \rho^{m}_i=\newfakerule{i}{m}{}{}$ for the ``brain''~$i$ and $\rho^{m}_j \in\{\newfreezerule{},\newffreezerule{j}{}\}$ for other $j\ne i$,  for $m=0, \dots,t-1$.
For any run $r \in \system{\chi}$,  all modified runs $r'\in R\bigl({\tauprotocol{S}{\envprotocol{}}{\joinprotocol{}^S}},{r},{\adj}\bigr)$ are $\tauprotocol{S}{\envprotocol{}}{\joinprotocol{}^S}$-transitional and satisfy the following properties:\looseness=-1
    \begin{compactenum}
       \item\label{lem:main_synch_compose-fake-freeze:i-same} 
       ``Brain'' agent $i$ cannot distinguish $r$ from $r'$: $\run[']{i}{m}=\run{i}{m}$ for all $m \le t$.
        \item\label{lem:main_synch_compose-fake-freeze:j-frozen} 
        Other agents $j\ne i$ remain in their initial states: $\run[']{j}{m}=\run[']{j}{0}$ for all $m \le t$.
		\item\label{lem:main_synch_compose-fake-freeze:i-bad} 
		Agent~$i$ is faulty from the beginning: $(i,m) \in \failed{r'}{t}$ for all $1 \leq m \leq t$.
		\item\label{lem:main_synch_compose-fake-freeze:only-i-failed} 
		Other agents $j \ne i$ are faulty by time $t$ if{f} $\rho^m_j = \newffreezerule{j}{}$ for some $m$.
    \end{compactenum}
\end{lemma}
\begin{proof}
	The proof is almost identical to that of the case  for 
	asynchronous agents from~\cite{KPSF19:FroCos}
. The only difference is in the proof that $r'$ is transitional.
The gullibility/delayability/fallibility assumptions ensure that the protocols can issue the sets of events prescribed by adjustment~$\adj$. By Def.~\ref{def:bitv_interv}, none  of the interventions $\newfakerule{i}{t}{}$, $\newfreezerule{}$, or $\newffreezerule{j}{}$
prescribes a $go$  before time~$t$. Thus, all actions are filtered out, and $\filtere[S]{}{}$ from Def.~\ref{def:filterSyn} does not remove any prescribed events.\looseness=-1 \qed
\end{proof}

\begin{lemma} \label{lem:synch_br_i_t_V_no_occur}
In the setting of Lemma~\ref{lem:synch_compose-fake-freeze}, no hap $o$ occurs correctly in any modified run $r' \in R\bigl({\tauprotocol{S}{\envprotocol{}}{\joinprotocol{}^S}},{r},{\adj}\bigr)$ before time~$t$, in other words, $(\I, r', m) \nvDash \trueoccurred{o}$ for all $m \leq t$.
\end{lemma} 
\begin{proof} Follows directly by unfolding Def.~\ref{def:bitv_interv} of the interventions $\newfakerule{i}{t}{}$, $\newfreezerule{}$, and $\newffreezerule{i}{}$
and the  definition of $\trueoccurred{o}$ 
	on p.~\pageref{page:occurred}.
\qed
\end{proof}

\begin{theorem}\label{lem:synch_no-k-occurred}
If the agent-context $\chi$ is non-excluding, the ``brain''~$i$ cannot know  
\begin{compactenum}[(a)]
\item \label{item:syn_intro_hap}
that any hap occurred correctly, i.e., 
$ \I \vDash \lnot\K{i}{\trueoccurred{o}}$ for all haps~$o$;
\item \label{item:syn_intro_self}
that it itself is correct, i.e., $\I \vDash \lnot K_i \correct{i}$;
\item \label{item:syn_intro_other_fault}
that another agent $j \ne i$ is faulty, i.e., $\I \vDash \lnot  \K{i}{\faulty{j}}$;
\item \label{item:syn_intro_other_correct}
that another agent $j \ne i$ is correct, i.e., $\I \vDash \lnot  \K{i}{\correct{j}}$.
\end{compactenum}
\end{theorem}
\begin{proof}
We need to show that all these knowledge statements are false  at $\kstruct{}{r}{t} $ for any $r \in R^\chi$ and any $t \in \bbbt$. 

For $t>0$,
	consider $\adj$ from Lemma~\ref{lem:synch_compose-fake-freeze} for this~$t$. It is possible to pick one modified run $r'\in R\bigl({\tauprotocol{S}{\envprotocol{}}{\joinprotocol{}^S}},{r},{\adj}\bigr)$ because $\chi$ is non-excluding. ``Brain'' $i$ cannot distinguish $r(t)$ from $r'(t)$ by Lemma~\ref{lem:synch_compose-fake-freeze}.\ref{lem:main_synch_compose-fake-freeze:i-same}. For Statement~\eqref{item:syn_intro_hap}, we have $(\I, r', t) \nvDash \trueoccurred{o}$ by Lemma~\ref{lem:synch_br_i_t_V_no_occur}. For Statement~\eqref{item:syn_intro_self}, $(\I, r', t) \nvDash \correct{i}$ by Lemma~\ref{lem:synch_compose-fake-freeze}.\ref{lem:main_synch_compose-fake-freeze:i-bad}. For Statement~\eqref{item:syn_intro_other_fault}, we have $(\I, r', t) \nvDash \faulty{j}$ if all $\rho^m_j=\newfreezerule{}$. Finally, for Statement~\eqref{item:syn_intro_other_correct}, we have $(\I, r', t) \nvDash \correct{j}$ if  $\rho^0_j=\newffreezerule{j}{}$. 
	
	It remains to consider the case of $t=0$. Here, $(\I, r, 0) \nvDash \trueoccurred{o}$ and $(\I, r, 0) \nvDash \faulty{j}$ trivially for the run $r$ itself, which completes the argument for  Statements~\eqref{item:syn_intro_hap} and~\eqref{item:syn_intro_other_fault}. However, since all agents are still correct at $t=0$, for Statements~\eqref{item:syn_intro_self} and~\eqref{item:syn_intro_other_correct}, we additionally notice ``brain''~$i$ is delayable because it is gullible. Since delaying~$i$ for the first round would prevent it from distinguishing whether $t=0$ or $t=1$, the case of $t=0$ is thereby reduced to the already considered case of $t>0$.
	\qed
%
\end{proof}

\begin{remark}
By contrast, 
 it is sometimes possible for synchronous agents to learn of their own defectiveness, i.e., $\I \nvDash \lnot K_i \faulty{i}$. This may happen, for instance, if there is a mismatch between actions recorded in the agent's local history and actions prescribed by the agent's protocol for the preceding local state.
\end{remark}

While the above limitations of knowledge apply to both asynchronous~\cite{KPSF19:FroCos} and synchronous agents, as we just showed, synchronous agents do gain awareness of the global clock in the following precise sense.

\begin{definition}
	For all $l \in \bbbn$, we add the following definition of truth for special propositional variables $\nsr{l}$ in interpreted systems $\intsys' = (\system{}', \interpretation{}{})$ (for $\system{}' \subseteq \system{}$):
$	(\intsys',r,t) \models \nsr{l}$ 
if{f}
$\NSR{r(t)} = l$.
\end{definition}

\begin{theorem}\label{lem:virtAware}
	A synchronous agent $k$ can always infer how many synced rounds elapsed from the beginning of a run under the assumption of its own correctness, i.e., for any run $r \in \system{\chi}$ and timestamp $t \in \bbbt$, we have
		$(\intsys,r,t) \models H_{k}{\nsr{\NSR{r(t)}}}$.
\end{theorem}
\begin{proof}
	Since $H_{k}{\nsr{\NSR{r(t)}}} = \correct{k} \to  \K{k}{\big(\correct{k} \rightarrow \nsr{\NSR{r(t)}}\big)}$, we need to show that $(\intsys,r',t') \models  \nsr{\NSR{r(t)}}$ whenever $r_k(t) = r'_k(t')$ and agent~$k$ is correct both at $r(t)$ and $r'(t')$.
It is not hard though tedious to prove that, for a correct agent, the number of $\clock$ actions in its local history is equal to the number of synced rounds elapsed in the run. Thus, $\NSR{r(t)}$ is equal to the number of $\clock$'s in $r_k(t) = r'_k(t')$, which, in turn, equals $\NSR{r'(t')}$.
\qed
\end{proof}


\section{The Extension Framework}\label{sec:extensions}
In this section, we present a glimpse into our modular extension framework, which augments the asynchronous byzantine framework \cite{KPSF19:FroCos} and enables us to implement and combine a variety of  system assumptions and, consequently, extend the epistemic analysis akin to that just performed for synchronous agents.

\begin{definition}[Extension]
	\label{def:extension}
	Let $\extension[\alpha] \ce (\mathit{PP}^\alpha,\mathit{IS}^\alpha,\transitionfrom[\alpha]{}{},\Admissibility[\alpha])$ with nonempty sets
	$\mathit{PP}^\alpha \subseteq \pprotocols$, $\mathit{IS}^\alpha \subseteq 2^{\globalstates(0)}$, and $\Admissibility[\alpha] \subseteq \system{}$  and a transition template~$\transitionfrom[\alpha]{}{}$. An agent-context~$\chi=\left((\envprotocol{},\globalinitialstates[\chi],\tau,\Psi),P\right)$ is part of $\extension[\alpha]$, denoted $\chi\in\extension[\alpha]$,  if{f} 
		$(\envprotocol{},\joinprotocol{})\in \mathit{PP}^\alpha$, 
		$\globalinitialstates[\chi] \in \mathit{IS}^\alpha$, 
		$\transitionfrom{}{} = \transitionfrom[\alpha]{}{}$, 
		$\Psi = \Admissibility[\alpha]$, 
		and
		$\System[\chi]\neq\varnothing$.
%
We call $\extension[\alpha]$ a \textbf{$($framework$)$ extension} 
 if{f} there exists an agent-context $\chi$ such that $\chi \in \extension[\alpha]$.\looseness=-1
%
%
\end{definition}

\ownsubsubsection{Extension combination.}
Combining extension requires combining their constituent parts. Since 
allowable pairs of protocols,  runs, and collections of initial states are restricted by an extension,  combining two extensions naturally means imposing both restrictions, i.e.,~taking their intersection.
Combining the respective transition templates imposed by these extensions, on the other hand, warrants more explanation.
Transition templates differ from each other only in the filtering phase. Therefore, combining  transition templates, in effect, means combining their respective filter functions, which can be done in various ways. In this section we discuss \emph{filter composition}.

\begin{definition}[Basic Filter Property] \label{def:basic_filter_prop}
	We call a function $\filtere[\alpha]{}{}$ ($\filterag[\alpha]{i}{}{}$ for $i \in \agents$) an \textit{event (action) filter function}  if{f}  
		$\filtere[\alpha]{h}{X_\epsilon,X_\agents} \subseteq X_\epsilon$ and  
		$\filterag[\alpha]{i}{X_\agents}{X_\epsilon} \subseteq X_i$ (for the exact function typification  see Appendix, Def.~\ref{def:filter}).
\end{definition}

\begin{definition}
    \label{def:envfilter_intersection}
	Given event (action) filter functions $\filtere[\alpha]{}{}$ and $\filtere[\beta]{}{}$ ($\filterag[\alpha]{i}{}{}$ and $\filterag[\beta]{i}{}{}$ for the same $i \in \agents$), their
	\textbf{filter composition} is defined as
	\begin{equation*} \label{equ:circ_filter}
		\begin{aligned}
			\filtere[\beta \circ \alpha]{h}{X_\epsilon,\, X_\agents} &\ce \filtere[\beta]{h}{\filtere[\alpha]{h}{X_\epsilon,\, X_\agents},\, X_\agents}, \\
			\filterag[\beta \circ \alpha]{i}{X_\agents}{X_\epsilon} &\ce 
			\filterag[\beta]{i}{X_{[1,i-1]},\, \filterag[\alpha]{i}{X_\agents\,}{X_\epsilon},\, X_{[i+1,n]}}{X_\epsilon}.
		\end{aligned}
	\end{equation*}
\end{definition}

\begin{definition}
\label{def:extension_intersection}
    For two extensions $\extension[\dag]=(\mathit{PP}^\dag,\mathit{IS}^\dag,\transitionfrom[\dag]{}{},\Admissibility[\dag])$ with $\dag \in \{\alpha,\beta\}$, 
    we define their
	composition  $\extension[\alpha\circ\beta] \ce (\mathit{PP}^\alpha \cap \mathit{PP}^\beta,\mathit{IS}^\alpha \cap \mathit{IS}^\beta,\transitionfrom[\alpha \circ \beta]{}{},\Admissibility[\alpha]\cap\Admissibility[\beta])$,
	where in $\transitionfrom[\alpha \circ \beta]{}{}$ the filters of $\transitionfrom[\alpha]{}{}$ and $\transitionfrom[\beta]{}{}$ are combined via filter composition (resulting in $\filtere[\alpha \circ \beta]{}{}$ and $\filterag[\alpha \circ \beta]{i}{}{}$ for each $i \in \agents$).
\end{definition}

Since such a combination $\extension[\alpha \circ \beta]$ may not  be a valid extension,  we introduce the notion of extension compatibility.
Informally, extensions are  \textbf{compatible} if their combination can produce runs (see Appendix, Def.~\ref{def:ext_compat}  for the details).\looseness=-1

%

We recall the definition of  the conventional (trace-based)  safety properties.\looseness=-1
\begin{definition}
\label{def:prefset_saf_prop}
	Let
		$\prefset \ce \{r(t) \mid  r \in \systrans,  t \in \bbbt\}\subseteq \globalstates$. 
	A nonempty set $S' \subseteq \systrans \sqcup \prefset$ is a \textbf{safety property} if
	\begin{compactenum} [(I)]
		\item $S'$ is \textbf{prefix-closed} in that
			\begin{compactitem}
			\item $r(t) \in S'$  implies that $r(t') \in S'$ for $t' \le t$ and
			\item $r' \in S'$ implies that $r'(t'') \in S'$ for all $t'' \in \bbbt$;
			\end{compactitem}
		\item $S'$ is \textbf{limit-closed}, i.e., $r(t) \in S'$ for  all $t \in \bbbt$ implies that $r \in S'$.
	\end{compactenum}
\end{definition}

For the formal reasoning that any property $P^\alpha \subseteq \systrans$ can be written as intersection of a safety and liveness property, see Appendix, Defs.~\ref{def:saf_liv_prep} and~\ref{def:trace_saf} and Lemmas~\ref{lem:L_is_live} and~\ref{lem:P_is_L_cap_S}.
Since safety properties based on traces are inconvenient for reasoning on a round by round basis, we introduce an equivalent safety property representation, better suited for this task.
The fact that our alternative safety property definition is indeed equivalent to the trace safety property representation, follows from Appendix, Def.~\ref{def:bij_map_tr_to_op} and Lemma \ref{lem:f_is_bijective}.

\begin{definition}
	\label{def:safety_property}
	An \textbf{operational safety property} $\saf$ is defined as a function
		$\saf \colon \prefset \to 2^{2^{\gevents} \times 2^{\gtrueactions[1]} \times \ldots \times 2^{\gtrueactions[n]}}$,
	which satisfies the following two conditions, called \emph{operational safety property attributes}. ($[]$~represents the empty sequence.)
	\begin{compactenum}
		\item\label{item:op_saf_attr_1} $(\exists h \in \prefset)\ h_\epsilon = [] \ \wedge \ S(h) \ne \varnothing$; 
		\item\label{item:op_saf_attr_2} $(\forall h \in \prefset)\ h_\epsilon \ne [] \ \rightarrow \\
		\big(\big((\exists h' \in \prefset)(\exists X \in S(h'))\ h = \update{h'}{X}\big) \ \leftrightarrow \ S(h) \ne \varnothing$\big). 
	\end{compactenum}
	The set of all operational safety properties is denoted by $\opsaf$.
\end{definition}
	Informally, Condition~\ref{item:op_saf_attr_1} means that there exists at least one safe initial state.
	Condition~\ref{item:op_saf_attr_2} means that every non-initial state is safely extendable if and only if it is safely reachable.
From this point on, whenever we refer to a safety property, we mean the operational safety property 
(the trace safety property representation can always be retrieved if desired).
%
We have discovered that \textbf{downward closure}  of the safety property of an extension greatly improves its composability. Fortunately, it turned out  that a few real-life safety properties (e.g., time-bounded communication, at-most-$f$ byzantine agents) are, in fact, downward closed.\looseness=-1

\begin{definition}
	A safety property $S$ is \textbf{downward closed} if{f}
		for all $h \in \prefset$,   $(X_\epsilon, X_\agents) \in \saf(h)$,  $X'_\epsilon \subseteq X_\epsilon$, and $X'_i \subseteq X_i$ for $i \in \agents$, we have  $(X'_\epsilon, X'_\agents) \in \saf(h)$.
\end{definition}

\ownsubsubsection{Implementation classes.}
According to Def.~\ref{def:extension}, specific system assumptions can be implemented via extensions using a combination of altering the set of environment protocols, set of agent protocols, event/action filter functions, and the admissibility condition.
One crucial question arises:
\emph{if a particular property could be implemented using different combinations of these mechanisms, which one of them should be favored?} \looseness=-1
The answer to this question is informed by our goal to construct extensions in the most  modular and composable manner.
Indeed, while the compatibility of two extensions guarantees their composition to produce runs, these runs may violate the safety property of one of the extensions, thereby defying the purpose of their combination. Here are two examples:
\begin{example} \label{ex:intuitive_ext_comb}
A necessary event whose presence is ensured by the protocol restrictions of one extension may be removed by the event filter of the other extension.	
\end{example}
\begin{example}
\label{ex:smuggle_cause}
Consider composing $\filtere[B]{}{}$ from~\cite{KPSF19:FroCos}, the \emph{causal filter} that removes any receive event without a matching send (correct or fake) (see~\eqref{eq:filter_function} for the formal definition),  with the \emph{synchronous agents filter} $\filtere[S]{}{}$  (see Def.~\ref{def:filterSyn}).
If  $\filtere[S]{}{}$ is applied last, it may remove some $go(i)$ event, preventing agent~$i$ from  sending a message necessary to support the causality of some receive event.
Thus, one must first apply $\filtere[S]{}{}$, followed by $\filtere[B]{}{}$.\looseness=-1
\end{example}
Therefore, in this section, we provide a classification of extension implementations, which we call \textbf{implementation classes}, in order to analyze their composability and answer our posed question.

\begin{table}[h] 
	\caption{Implementation classes}
	\label{tab:impl_classes}
	\resizebox{\textwidth}{!}{
	\begin{tabular} {| >{\centering}p{2.4cm} | >{\centering\arraybackslash}m{1.2cm} | >{\centering\arraybackslash}m{1.2cm} | >{\centering\arraybackslash}m{1.2cm} | >{\centering\arraybackslash}m{1.2cm} | >{\centering\arraybackslash}m{1.2cm} | >{\centering\arraybackslash}m{1.2cm} | >{\centering\arraybackslash}m{1.2cm} | >{\centering\arraybackslash}m{1.2cm} | >{\centering\arraybackslash}m{1.2cm} |}
	\hline
		\tiny $\ic$	&\tiny\textbf{admiss. condition}	&\tiny\textbf{initial states}	&\tiny\textbf{joint protocols}	&\tiny\textbf{environ. protocols}	&\tiny\textbf{arbitrary event filter}	&\tiny\textbf{standard action filters 
		}	&\tiny\textbf{arbitrary action filters}	&\tiny\textbf{downward closed}	&\tiny\textbf{monotonic filters} \tabularnewline \hline
		\tiny$\adm$           	&x&x&&&&&&& \tabularnewline \hline
		\tiny$\jp$	            &x&x&x&&&&&& \tabularnewline \hline
		\tiny$\jpfb$            &x&x&x&&&x&&& \tabularnewline \hline
		\tiny$\ejp$	 	      	&x&x&x&x&&&&& \tabularnewline \hline
		\tiny$\ejpfb$	        &x&x&x&x&&x&&& \tabularnewline \hline
		\tiny$\efjp$	       	&x&x&x&&x&&&& \tabularnewline \hline
		\tiny$\efjpfb$	        &x&x&x&&x&x&&& \tabularnewline \hline
		\tiny$\efejp$	        &x&x&x&x&x&&&& \tabularnewline \hline
		\tiny$\efejpfb$	        &x&x&x&x&x&x&&& \tabularnewline \hline
		\tiny$\others$	        &x&x&x&x&x&&x&& \tabularnewline \hline
		\tiny$\jpdc$	       	&x&x&x&&&&&x& \tabularnewline \hline
		\tiny$\ejpdc$	        &x&x&x$^{1}$&x$^{1}$&&&&x& \tabularnewline \hline
		\tiny$\efjpdc$	        &x&x&x$^{1}$&&x$^{1}$&&&x& \tabularnewline \hline
		\tiny$\efejpdc$	        &x&x&x$^{1}$&x$^{1}$&x$^{1}$&&&x& \tabularnewline \hline
		\tiny$\othersdc$	   	&x&x&x$^{1}$&x$^{1}$&x$^{1}$&&x$^1$&x& \tabularnewline \hline
		\tiny$\efejpdcmono$	    &x&x&x$^{1}$&x$^{1}$&x$^{1,2}$&&&x&x \tabularnewline \hline
		\tiny$\othersdcmono$  	&x&x&x$^{1}$&x$^{1}$&x$^{1,2}$&&x$^{1,2}$&x&x \tabularnewline \hline
	\end{tabular}
	}
\tiny
1) such that the extension's safety property remains downward closed \\
2) such that the extension's filters are monotonic
\end{table}
\begin{definition} \label{def:impl_class}
	\textbf{Implementation classes} are sets of extensions 
	presented in Table~\ref{tab:impl_classes}, where the name of the implementation class is stated in the leftmost column and  parts manipulated  and properties satisfied 
 by this extension are marked by ``\textup{x}'' in its row.
	Note that the last seven~classes are subsets of other classes. We  consider them separately due to their altered attributes regarding composability.
	The set of all implementation classes is denoted by~$\ic$.%
\footnote{For a detailed definition see Appendix, Def.~\ref{def:impl_classes_extended}.}\looseness=-1
\end{definition}

To describe  implementation class \textbf{composability}, we introduce the forth and reverse composability relations.\looseness=-1

\begin{definition}
Two implementation classes $IC^\alpha, IC^\beta \in \ic$ are \textbf{forth} (\textbf{reverse}) \textbf{composable} if{f} for all extensions $\extension[\alpha] \in IC^\alpha$ and $\extension[\beta] \in IC^\beta$  compatible with respect to forth (reverse) composition $\alpha \circ \beta$ ($\beta \circ \alpha$), the extension $\extension[\alpha \circ \beta]$ ($\extension[\beta \circ \alpha]$) adheres to the safety property $S^\beta$ of $\extension[\beta]$.\footnote{Table~\ref{tab:composability_matrix} states  our composability results for various implementation class combinations.}
\looseness=-1
\end{definition}



Our synchronous agents introduced in Sect.~\ref{sec:synchronous_agents} correspond to the following extension from the class $\efjpfb$:
\begin{definition}
    \label{def:synchronous_a_extension}
    \index{$\extension[S]$}
    We denote by
		$\extension[S] \ce \left(\envprotocols\times\agprotocols[S], 2^{\globalstates(0)} \setminus \{\varnothing\}, \transitionfrom[S]{}{}, \System\right)$
	the \textbf{synchronous agents extension}, where the transition template $\transitionfrom[S]{}{}$ uses the synchronous agents event filter 
	and the standard action filters
	.\looseness=-1
\end{definition}

\begin{table}[t] 
	\caption{Composability matrix of implementation classes }
	\label{tab:composability_matrix}
	\resizebox{\textwidth}{!}{
	\begin{tabular} {| p{2.4cm} | >{\centering\arraybackslash}m{0.6cm} | >{\centering\arraybackslash}m{0.6cm} | >{\centering\arraybackslash}m{0.6cm} | >{\centering\arraybackslash}m{0.6cm} | >{\centering\arraybackslash}m{0.6cm} | >{\centering\arraybackslash}m{0.6cm} | >{\centering\arraybackslash}m{0.6cm} | >{\centering\arraybackslash}m{0.6cm} | >{\centering\arraybackslash}m{0.6cm} | >{\centering\arraybackslash}m{0.6cm} | >{\centering\arraybackslash}m{0.6cm} | >{\centering\arraybackslash}m{0.6cm} | >{\centering\arraybackslash}m{0.6cm} | >{\centering\arraybackslash}m{0.6cm} | >{\centering\arraybackslash}m{0.6cm} | >{\centering\arraybackslash}m{0.6cm} | >{\centering\arraybackslash}m{0.6cm} |}
	\hline
							&\tiny\textbf{Adm}	&\tiny\textbf{JP}	&\tiny\textbf{Env JP}	&\tiny\textbf{EvF JP}	&\tiny\textbf{EvF Env JP}	&\tiny\textbf{JP - AFB}	&\tiny\textbf{Env JP - AFB}	&\tiny\textbf{EvF JP - AFB}	&\tiny\textbf{EvF Env JP - AFB}	&\tiny\textbf{Oth ers}	&\tiny\textbf{JP $_{DC}$}	&\tiny\textbf{Env JP $_{DC}$}	&\tiny\textbf{EvF JP $_{DC}$}	&\tiny\textbf{EvF Env JP $_{DC}$}	&\tiny\textbf{Oth ers $_{DC}$}	&\tiny\textbf{EvF Env JP $_{DC}$ $_{mono}$}	&\tiny\textbf{Oth ers $_{DC}$ $_{mono}$} \tabularnewline
		\hline
		\tiny$\adm$           	&\cellcolor{green}c	&\cellcolor{green}c	&\cellcolor{green}c	&\cellcolor{green}c	&\cellcolor{green}c	&\cellcolor{green}c	&\cellcolor{green}c	&\cellcolor{green}c	&\cellcolor{green}c	&\cellcolor{green}c	&\cellcolor{green}c	&\cellcolor{green}c	&\cellcolor{green}c	&\cellcolor{green}c	&\cellcolor{green}c	&\cellcolor{green}c	&\cellcolor{green}c \tabularnewline \hline
		\tiny$\jp$	            &\cellcolor{green}c	&\cellcolor{green}c	&\cellcolor{green}c	&\cellcolor{green}c	&\cellcolor{green}c	&\cellcolor{green}c	&\cellcolor{green}c	&\cellcolor{green}c	&\cellcolor{green}c	&\cellcolor{green}c	&\cellcolor{green}c	&\cellcolor{green}c	&\cellcolor{green}c	&\cellcolor{green}c	&\cellcolor{green}c	&\cellcolor{green}c	&\cellcolor{green}c \tabularnewline \hline
		\tiny$\ejp$	            &\cellcolor{green}c	&\cellcolor{green}c	&\cellcolor{green}c	&\cellcolor{green}c	&\cellcolor{green}c	&\cellcolor{green}c	&\cellcolor{green}c	&\cellcolor{green}c	&\cellcolor{green}c	&\cellcolor{green}c	&\cellcolor{green}c	&\cellcolor{green}c	&\cellcolor{green}c	&\cellcolor{green}c	&\cellcolor{green}c	&\cellcolor{green}c	&\cellcolor{green}c \tabularnewline \hline
		\tiny$\efjp$	       	&\cellcolor{green}c	&\cellcolor{green}c	&\cellcolor{red!50}		&\cellcolor{cyan}r	&\cellcolor{red!50}	&\cellcolor{green}c	&\cellcolor{red!50}	&\cellcolor{cyan}r	&\cellcolor{red!50}	&\cellcolor{red!50}	&\cellcolor{green}c	&\cellcolor{green}c	&\cellcolor{green}c	&\cellcolor{blue}{\color{white}f}	&\cellcolor{blue}{\color{white}f}	&\cellcolor{green}c	&\cellcolor{green}c \tabularnewline \hline
		\tiny$\efejp$	        &\cellcolor{green}c	&\cellcolor{green}c	&\cellcolor{red!50}		&\cellcolor{cyan}r	&\cellcolor{red!50}	&\cellcolor{green}c	&\cellcolor{red!50}	&\cellcolor{cyan}r	&\cellcolor{red!50}	&\cellcolor{red!50}	&\cellcolor{green}c	&\cellcolor{green}c	&\cellcolor{green}c	&\cellcolor{blue}{\color{white}f}	&\cellcolor{blue}{\color{white}f}	&\cellcolor{green}c	&\cellcolor{green}c \tabularnewline \hline
		\tiny$\jpfb$	       	&\cellcolor{green}c	&\cellcolor{red!50}		&\cellcolor{red!50}		&\cellcolor{red!50}		&\cellcolor{red!50}	&\cellcolor{green}c	&\cellcolor{green}c	&\cellcolor{green}c	&\cellcolor{green}c	&\cellcolor{red!50}	&\cellcolor{green}c	&\cellcolor{green}c	&\cellcolor{green}c	&\cellcolor{green}c	&\cellcolor{blue}{\color{white}f}	&\cellcolor{green}c	&\cellcolor{green}c \tabularnewline \hline
		\tiny$\ejpfb$	        &\cellcolor{green}c	&\cellcolor{red!50}		&\cellcolor{red!50}		&\cellcolor{red!50}		&\cellcolor{red!50}	&\cellcolor{green}c	&\cellcolor{green}c	&\cellcolor{green}c	&\cellcolor{green}c	&\cellcolor{red!50}	&\cellcolor{green}c	&\cellcolor{green}c	&\cellcolor{green}c	&\cellcolor{green}c	&\cellcolor{blue}{\color{white}f}	&\cellcolor{green}c	&\cellcolor{green}c \tabularnewline \hline
		\tiny$\efjpfb$	        &\cellcolor{green}c	&\cellcolor{red!50}		&\cellcolor{red!50}		&\cellcolor{red!50}		&\cellcolor{red!50}	&\cellcolor{green}c	&\cellcolor{red!50}	&\cellcolor{cyan}r	&\cellcolor{red!50}	&\cellcolor{red!50}	&\cellcolor{green}c	&\cellcolor{green}c	&\cellcolor{green}c	&\cellcolor{blue}{\color{white}f}	&\cellcolor{blue}{\color{white}f}	&\cellcolor{green}c	&\cellcolor{green}c \tabularnewline \hline
		\tiny$\efejpfb$	        &\cellcolor{green}c	&\cellcolor{red!50}		&\cellcolor{red!50}		&\cellcolor{red!50}		&\cellcolor{red!50}	&\cellcolor{green}c	&\cellcolor{red!50}	&\cellcolor{cyan}r	&\cellcolor{red!50}	&\cellcolor{red!50}	&\cellcolor{green}c	&\cellcolor{green}c	&\cellcolor{green}c	&\cellcolor{blue}{\color{white}f}	&\cellcolor{blue}{\color{white}f}	&\cellcolor{green}c	&\cellcolor{green}c \tabularnewline \hline
		\tiny$\others$	        &\cellcolor{green}c	&\cellcolor{red!50}		&\cellcolor{red!50}		&\cellcolor{red!50}		&\cellcolor{red!50}	&\cellcolor{red!50}	&\cellcolor{red!50}	&\cellcolor{red!50}	&\cellcolor{red!50}	&\cellcolor{red!50}	&\cellcolor{green}c	&\cellcolor{green}c	&\cellcolor{green}c	&\cellcolor{blue}{\color{white}f}	&\cellcolor{blue}{\color{white}f}	&\cellcolor{green}c	&\cellcolor{green}c \tabularnewline \hline
		\tiny$\jpdc$	       	&\cellcolor{green}c	&\cellcolor{green}c	&\cellcolor{green}c	&\cellcolor{green}c	&\cellcolor{green}c	&\cellcolor{green}c	&\cellcolor{green}c	&\cellcolor{green}c	&\cellcolor{green}c	&\cellcolor{green}c	&\cellcolor{green}c	&\cellcolor{green}c	&\cellcolor{green}c	&\cellcolor{green}c	&\cellcolor{green}c	&\cellcolor{green}c	&\cellcolor{green}c \tabularnewline \hline
		\tiny$\ejpdc$	        &\cellcolor{green}c	&\cellcolor{green}c	&\cellcolor{green}c	&\cellcolor{green}c	&\cellcolor{green}c	&\cellcolor{green}c	&\cellcolor{green}c	&\cellcolor{green}c	&\cellcolor{green}c	&\cellcolor{green}c	&\cellcolor{green}c	&\cellcolor{green}c	&\cellcolor{green}c	&\cellcolor{green}c	&\cellcolor{green}c	&\cellcolor{green}c	&\cellcolor{green}c \tabularnewline \hline
		\tiny$\efjpdc$	        &\cellcolor{green}c	&\cellcolor{green}c	&\cellcolor{red!50}		&\cellcolor{cyan}r	&\cellcolor{red!50}		&\cellcolor{green}c	&\cellcolor{red!50}	&\cellcolor{cyan}r	&\cellcolor{red!50}	&\cellcolor{red!50}	&\cellcolor{green}c	&\cellcolor{green}c	&\cellcolor{green}c	&\cellcolor{blue}{\color{white}f}	&\cellcolor{blue}{\color{white}f}	&\cellcolor{green}c	&\cellcolor{green}c \tabularnewline \hline
		\tiny$\efejpdc$	        &\cellcolor{green}c	&\cellcolor{green}c	&\cellcolor{red!50}		&\cellcolor{cyan}r	&\cellcolor{red!50}		&\cellcolor{green}c	&\cellcolor{red!50}	&\cellcolor{cyan}r	&\cellcolor{red!50}	&\cellcolor{red!50}	&\cellcolor{green}c	&\cellcolor{green}c	&\cellcolor{green}c	&\cellcolor{blue}{\color{white}f}	&\cellcolor{blue}{\color{white}f}	&\cellcolor{green}c	&\cellcolor{green}c \tabularnewline \hline
		\tiny$\othersdc$	   	&\cellcolor{green}c	&\cellcolor{red!50}		&\cellcolor{red!50}		&\cellcolor{red!50}		&\cellcolor{red!50}		&\cellcolor{red!50}	&\cellcolor{red!50}	&\cellcolor{red!50}	&\cellcolor{red!50}	&\cellcolor{red!50}	&\cellcolor{green}c	&\cellcolor{green}c	&\cellcolor{green}c	&\cellcolor{blue}{\color{white}f}	&\cellcolor{blue}{\color{white}f}	&\cellcolor{green}c	&\cellcolor{green}c \tabularnewline \hline
		\tiny$\efejpdcmono$	    &\cellcolor{green}c	&\cellcolor{green}c	&\cellcolor{red!50}		&\cellcolor{cyan}r	&\cellcolor{red!50}		&\cellcolor{green}c	&\cellcolor{red!50}	&\cellcolor{cyan}r	&\cellcolor{red!50}	&\cellcolor{red!50}	&\cellcolor{green}c	&\cellcolor{green}c	&\cellcolor{green}c	&\cellcolor{blue}{\color{white}f}	&\cellcolor{blue}{\color{white}f}	&\cellcolor{green}c	&\cellcolor{green}c \tabularnewline \hline
		\tiny$\othersdcmono$  	&\cellcolor{green}c	&\cellcolor{red!50}		&\cellcolor{red!50}		&\cellcolor{red!50}		&\cellcolor{red!50}		&\cellcolor{red!50}	&\cellcolor{red!50}	&\cellcolor{red!50}	&\cellcolor{red!50}	&\cellcolor{red!50}	&\cellcolor{green}c	&\cellcolor{green}c	&\cellcolor{green}c	&\cellcolor{blue}{\color{white}f}	&\cellcolor{blue}{\color{white}f}	&\cellcolor{green}c	&\cellcolor{green}c \tabularnewline \hline
	\end{tabular}
	}
\end{table}
\noindent
The entries in Table \ref{tab:composability_matrix} are to be read as follows supposing $x$ is the content of the entry, $LC$ is the implementation class to the left and $TC$ is the implementation class on the top:
\begin{compactitem}
	\item $x=c$ means that $LC$ is both forth and reverse composable with $TC$.
	\item $x=f$ means that $LC$ is forth composable with $TC$.
	\item $x=r$ means that $LC$ is reverse composable with $TC$.
	\item An empty entry means that $LC$ can generally not be safely combined with $TC$ (we do not have a positive result stating the opposite).
\end{compactitem}


\section{Lock-step Synchronous Agents} \label{sec:lss_agents}

In lock-step synchronous distributed systems \cite{lynch1996distributed}, agents act synchronously in \emph{communication-closed} rounds.
In each such round, every correct agent sends a message to every agent, which is received in the same round, and finally processes all received messages, which happens simultaneously at all correct agents.
Thus, agents are not only synchronous, but additionally their communication is reliable, broadcast, synchronous (and causal).
Our \emph{lock-step round extension}  combines 5~different extensions corresponding to the aforementioned properties: the
(i)~byzantine agents \cite{KPSF19:FroCos}, (ii)~synchronous agents extension from Def.~\ref{def:synchronous_a_extension}, (iii)~reliable communication extension ensuring that every sent message is eventually delivered, (iv)~synchronous communication extension ensuring that every message is either received instantaneously or not at all, and (v)~broadcast communication extension ensuring that every correctly sent message is sent to all agents.\footnote{Formal definitions for~(i), (iii), (iv),~(v) can be found in  Appendix, starting from Def.~\ref{def:asynch_byz_ag_ext}. }
\looseness=-1

Since, by Lemma \ref{lem:synch_compose-fake-freeze}, even synchronous agents can be fooled by their own (faulty) imagination, it is natural to ask whether a brain-in-a-vat scenario is still possible in the more restricted lock-step synchronous setting.
The proof of the possibility of the brain-in-a-vat scenario from Lemma~\ref{lem:synch_compose-fake-freeze} in an asynchronous setting provided in~\cite{KPSF19:FroCos} suggests this not to be the case. However, by considering extension combinations more closely and in more detail
we were able to implement such a scenario despite the additional restrictions. The issue is that the $i$-intervention $\newfakerule{i}{t}{}$ from Def.~\ref{def:bitv_interv} makes it possible for byzantine actions of the dreaming ``brain'' to affect other agents. This possibility becoming a certainty due to the more reliable communication of lock-step synchronous agents is the obstacle preventing the complete isolation of the brain. \looseness=-1
 This effect can be avoided by  modifying $\newfakerule{i}{t}{}$ to make all byzantine actions entirely imaginary. This new $i$-intervention $\improvednewfakerule{i}{t}{}$ is obtained from Def.~\ref{def:bitv_interv} by replacing $\betab[i]{r}{t}$ in the modified events with:
\begin{gather*} 
		\left\{\fakeof{i}{E} \mid \fakeof{i}{E} \in \betab[i]{r}{t}  \right\} \quad \cup \\
		\left\{\fakeof{i}{\mistakefor{\tick}{A}} \mid (\exists A' \in \gactions[i] \sqcup \{\tick\})  \fakeof{i}{\mistakefor{A'}{A}} \in \betab[i]{r}{t}  \right\}.
\end{gather*}
(The full reformulation of Lemma~\ref{lem:synch_compose-fake-freeze} for this case with a short proof can be found in  Appendix, Lemma \ref{lem:lock_step_synch_compose-fake-freeze}.) Therefore, even perfect clocks and communication-closed rounds do not exclude the ``brain-in-a-vat'' scenario, with the consequence that most (negative) introspection results for synchronous systems also hold for lock-step synchronous systems:
\begin{theorem}
Replacing $\newfakerule{i}{t}{}$   with $\improvednewfakerule{i}{t}{r}$ in Lemma~\ref{lem:synch_compose-fake-freeze} extends the latter's ``brain-in-the-vat'' properties to lock-step synchronous system.
\end{theorem}

But besides the fact that our lock-step round extension was instrumental for identifying the subtle improvements in implementing the brain-in-the-vat scenario, it does have positive consequences for the fault-detection abilities of the agents as well: using a weaker epistemic notion of the hope modality $H$, we have shown that in a lock-step synchronous context  it is possible to design  agent protocols to detect faults of other agents.

\begin{theorem} \label{thm:lock_step_synch_gl_int_K_faulty}
	There exists an agent context $\chi \in \extension[\mathit{LSS}]$, where $\chi = \bigl((\envprotocol{}^{\mathit{SC}_{\agents^2}}, \allowbreak \globalinitialstates, \allowbreak \transitionfrom[B \circ S]{}{}, \allowbreak \EDel_{\agents^2}),\widetilde{\joinprotocol{}}^{\mathit{SMC}_{\mathit{BCh}}}\bigr)$,
	and a run $r \in \system{\chi}$, such that
	for agents $i,j \in \agents$, where $i \ne j$,
	some timestamp $t \in \mathbb{N}$,
	and a $\chi$-based interpreted system $\intsys = (\system{\chi}, \interpretation{}{})$
	\begin{equation*} \label{equ:lock_step_synch_gl_int_K_faulty}
		\kstruct{}{r}{t} \models H_{i}{\faulty{j}}.
	\end{equation*}
\end{theorem}
\begin{proof}
	See Appendix, Theorem \ref{thm:lock_step_synch_gl_int_K_faulty_appendix}.
\end{proof}
\section{Conclusions}
\label{sec:conclusions}

We substantially augmented the epistemic reasoning framework for byzantine distributed systems \cite{KPSF19:FroCos} 
with extensions, which allow to incorporate additional system assumptions in a modular fashion.
By instantiating our extension framework for both synchronous and lock-step synchronous systems,
we proved that even adding perfect clocks and communication-closed rounds cannot circumvent the possibility of a brain-in-the-vat 
scenario and resulting negative introspection results, albeit they do enable some additional fault detection capabilities.\looseness=-1

%
%
%
\bibliographystyle{abbrvurl}
\bibliography{bib}

\begin{thebibliography}{10}

\bibitem{BM13:ICLA}
I.~Ben-Zvi and Y.~Moses.
\newblock Agent-time epistemics and coordination.
\newblock In {\em {ICLA~2013}}, volume 7750 of {\em LNCS}, pages 97--108.
  Springer, 2013.
\newblock \href {http://dx.doi.org/10.1007/978-3-642-36039-8\_9}
  {\path{doi:10.1007/978-3-642-36039-8\_9}}.

\bibitem{ben2014beyond}
I.~Ben-Zvi and Y.~Moses.
\newblock Beyond {L}amport's \it happened-before\rm: On time bounds and the
  ordering of events in distributed systems.
\newblock {\em Journal of the ACM}, 61(2:13), 2014.
\newblock \href {http://dx.doi.org/10.1145/2542181}
  {\path{doi:10.1145/2542181}}.

\bibitem{CGM14}
A.~Casta{\~n}eda, Y.~A. Gonczarowski, and Y.~Moses.
\newblock Unbeatable consensus.
\newblock In {\em {DISC~2014}}, volume 8784 of {\em LNCS}, pages 91--106.
  Springer, 2014.
\newblock \href {http://dx.doi.org/10.1007/978-3-662-45174-8\_7}
  {\path{doi:10.1007/978-3-662-45174-8\_7}}.

\bibitem{dwork1990knowledge}
C.~Dwork and Y.~Moses.
\newblock Knowledge and common knowledge in a {B}yzantine environment: Crash
  failures.
\newblock {\em Information and Computation}, 88:156--186, 1990.
\newblock \href {http://dx.doi.org/10.1016/0890-5401(90)90014-9}
  {\path{doi:10.1016/0890-5401(90)90014-9}}.

\bibitem{bookof4}
R.~Fagin, J.~Y. Halpern, Y.~Moses, and M.~Y. Vardi.
\newblock {\em Reasoning About Knowledge}.
\newblock MIT Press, 1995.

\bibitem{FLP85}
M.~J. Fischer, N.~A. Lynch, and M.~S. Paterson.
\newblock Impossibility of distributed consensus with one faulty process.
\newblock {\em Journal of the ACM}, 32:374--382, 1985.
\newblock \href {http://dx.doi.org/10.1145/3149.214121}
  {\path{doi:10.1145/3149.214121}}.

\bibitem{GM13:TARK}
Y.~A. Gonczarowski and Y.~Moses.
\newblock Timely common knowledge.
\newblock In {\em {TARK~XIV}}, pages 79--93, 2013.
\newblock URL: \url{https://arxiv.org/abs/1310.6414}.

\bibitem{GM18:PODC}
G.~Goren and Y.~Moses.
\newblock Silence.
\newblock In {\em {PODC~'18}}, pages 285--294. ACM, 2018.
\newblock \href {http://dx.doi.org/10.1145/3212734.3212768}
  {\path{doi:10.1145/3212734.3212768}}.

\bibitem{HM90}
J.~Y. Halpern and Y.~Moses.
\newblock Knowledge and common knowledge in a distributed environment.
\newblock {\em Journal of the ACM}, 37:549--587, 1990.
\newblock \href {http://dx.doi.org/10.1145/79147.79161}
  {\path{doi:10.1145/79147.79161}}.

\bibitem{halpern2001characterization}
J.~Y. Halpern, Y.~Moses, and O.~Waarts.
\newblock A characterization of eventual {B}yzantine agreement.
\newblock {\em SIAM Journal on Computing}, 31:838--865, 2001.
\newblock \href {http://dx.doi.org/10.1137/S0097539798340217}
  {\path{doi:10.1137/S0097539798340217}}.

\bibitem{KPSF19:TARK}
R.~Kuznets, L.~Prosperi, U.~Schmid, and K.~Fruzsa.
\newblock Causality and epistemic reasoning in byzantine multi-agent systems.
\newblock In {\em {TARK~2019}}, volume 297 of {\em {EPTCS}}, pages 293--312,
  2019.
\newblock \href {http://dx.doi.org/10.4204/EPTCS.297.19}
  {\path{doi:10.4204/EPTCS.297.19}}.

\bibitem{KPSF19:FroCos}
R.~Kuznets, L.~Prosperi, U.~Schmid, and K.~Fruzsa.
\newblock Epistemic reasoning with byzantine-faulty agents.
\newblock In {\em {FroCoS~2019}}, volume 11715 of {\em {LNCS}}, pages 259--276.
  Springer, 2019.
\newblock \href {http://dx.doi.org/10.1007/978-3-030-29007-8_15}
  {\path{doi:10.1007/978-3-030-29007-8_15}}.

\bibitem{Lam78}
L.~Lamport.
\newblock Time, clocks, and the ordering of events in a distributed system.
\newblock {\em Communications of the ACM}, 21:558--565, 1978.
\newblock \href {http://dx.doi.org/10.1145/359545.359563}
  {\path{doi:10.1145/359545.359563}}.

\bibitem{lamport1982byzantine}
L.~Lamport, R.~Shostak, and M.~Pease.
\newblock The {B}yzantine {G}enerals {P}roblem.
\newblock {\em ACM Transactions on Programming Languages and Systems},
  4:382--401, 1982.
\newblock \href {http://dx.doi.org/10.1145/357172.357176}
  {\path{doi:10.1145/357172.357176}}.

\bibitem{lynch1996distributed}
N.~A. Lynch.
\newblock {\em Distributed Algorithms}.
\newblock Morgan Kaufmann, 1996.

\bibitem{moses1986programming}
Y.~Moses and M.~R. Tuttle.
\newblock Programming simultaneous actions using common knowledge: Preliminary
  version.
\newblock In {\em {FOCS~1986}}, pages 208--221. IEEE, 1986.
\newblock \href {http://dx.doi.org/10.1109/SFCS.1986.46}
  {\path{doi:10.1109/SFCS.1986.46}}.

\bibitem{MT88}
Y.~Moses and M.~R. Tuttle.
\newblock Programming simultaneous actions using common knowledge.
\newblock {\em Algorithmica}, 3:121--169, 1988.
\newblock \href {http://dx.doi.org/10.1007/BF01762112}
  {\path{doi:10.1007/BF01762112}}.

\end{thebibliography}

\newpage
\appendix
\section{Appendix} \label{sec:appendix}

\begin{definition}
\label{def:filter}
    \index{${\filtere{}{}}$}
	We define an \textbf{event filter function}
	\begin{equation*}
		\filtere{}{}\colon\globalstates \times 2^{\gevents} \times 2^{\gtrueactions[1]} \times \dots \times 2^{\gtrueactions[n]} \longrightarrow  2^{\gevents}.
	\end{equation*}

    In addition, we define  \textbf{action filter functions} for agents $i \in\agents$
	\begin{equation*}
		\filterag{i}{}{} : 2^{\gtrueactions[1]} \times \dots \times 2^{\gtrueactions[n]} \times2^{\gevents} \longrightarrow 2^{\gtrueactions[i]}.
	\end{equation*}
\end{definition}

\begin{figure*}[h]
    \begin{center}
        \scalebox{0.50}{

\tikzset{nodes0/.style={anchor=west},}
\tikzset{nodes1/.style={text width=1.5cm,anchor=west},}
\tikzset{nodes2a/.style={text width=0.5cm,anchor=west},}
\tikzset{nodes2/.style={text width=0.8cm,anchor=west},}
\tikzset{nodes3/.style={text width=1.2cm,anchor=west},}
\tikzset{nodes4/.style={text width=1.3cm,anchor=west},}

\begin{tikzpicture}
    \node[nodes0] at (0,1) (0) {$\run{}{t}$};

    \node[nodes1] at (2.5,3)  (11) {$\agprotocol{n}{\run{n}{t}}$};
    \node[nodes1] at (2.5,2)  (01) {$\dots$};
    \node[nodes1] at (2.5,1)  (1) {$\agprotocol{1}{\run{1}{t}}$};
    \node[nodes1] at (2.5,-2) (2) {$\envprotocol{t}$};

    \draw[->,fill=white] (0) -- (11) node[midway,fill=white]{$\agprotocol{n}{}$};
    \draw[->] (0) -- (1) node[midway,fill=white]{$\agprotocol{1}{}$};
    \draw[->] (0) -- (2) node[midway,fill=white]{$\envprotocol{}$};

    \node[nodes2a] at (7,3)  (13a) {$X_n$};
    \node[nodes2a] at (7,2)  (03a) {$\dots$};
    \node[nodes2a] at (7,1)  (3a) {$X_1$};
    \node[nodes2a] at (7,-2) (4a) {$X_\epsilon$};
    
    \draw[->] (11) -- (13a) node[midway,fill=white] {$\adversary{}$};
    \draw[->] (1) -- (3a) node[midway,fill=white] {$\adversary{}$};
    \draw[->] (2) -- (4a) node[midway,fill=white] {$\adversary{}$};
    
    \node[nodes2] at (9.5,3)  (13) {$\alphaag{n}{r}{t}$};
    \node[nodes2] at (9.5,2)  (03) {$\dots$};
    \node[nodes2] at (9.5,1)  (3) {$\alphaag{1}{r}{t}$};
    \node[nodes2] at (12.5,3)  (13m) {$\alphaag{n}{r}{t}$};
    \node[nodes2] at (12.5,2)  (03m) {$\dots$};
    \node[nodes2] at (12.5,1)  (3m) {$\alphaag{1}{r}{t}$};

    \node[nodes2] at (9.5,-2) (4) {$X_\epsilon= \alphae{r}{t}$};
    
    \draw[->] (13a) -- (13) node[midway,fill=white] {$\mathit{global}$};
    \draw[->] (3a) -- (3) node[midway,fill=white] {$\mathit{global}$};
    \draw[double] (4a) -- (4) node[midway] {};
    \draw[double] (4) -- (4) node[midway] {};
    \node[nodes3] at (16,3) (15) {$\betaag{n}{r}{t}$};
    \node[nodes3] at (16,2) (05) {$\dots$};
    \node[nodes3] at (16,1) (5) {$\betaag{1}{r}{t}$};
    \node[nodes3] at (12.5,-2) (6m) {$\betae{r}{t}$};
    \node[nodes3] at (16,-2) (6) {$\betae{r}{t}$};
    
    \draw[double] (13) -- (13m);
    \draw[double] (3) -- (3m);
    \draw[double] (6) -- (6m);

    \draw[->] (13m) -- (15) node[midway,fill=white] (f3)  {$\filterag{n}{}{}$};
    \draw[->] (3m) -- (5) node[midway,fill=white] (f1) {$\filterag{1}{}{}$};
    \draw[->] (4) -- (6m) node[midway,fill=white] (fe) {$\filtere{}{}$};

    \draw[->] (13) edge (fe) (3) edge (fe);
    \draw[->,dotted] (13) edge (f1) (3) edge (f3);
    \draw[->,dashed] (6m) edge (f1) (6m) edge (f3);    
    
    \node[nodes4] at (19,3) (17) {$\run{n}{t+1}$};
    \node[nodes4] at (19,2) (07) {$\dots$};
    \node[nodes4] at (19,1) (7) {$\run{1}{t+1}$};
    \node[nodes4] at (19,-2) (8) {$\run{\epsilon}{t+1}$};

    \draw[->] (15) -- (17) node[midway,fill=white](19) {$\updateag{n}{}{}{}$};
    \draw[->,fill=white] (5) -- (7) node[midway,fill=white] (9){$\updateag{1}{}{}{}$};
    \draw[->] (6) -- (8) node[midway,fill=white] (a){$\updatee{}{}$};
    
    \draw[->] (5) -- (a);
    \draw[->] (15) -- (a);
    
    \draw[->,dashed] (6) -- (19) node[near end,fill=white] {$\betae[n]{r}{t}$};    
    \draw[->,dashed] (6) -- (9) node[near end,right,fill=white] {$\betae[1]{r}{t}$};

    \node[nodes0] at (21.5,1) (10){$\run{}{t+1}$};
    \draw[->] (7) edge (10) (17) edge (10) (8) edge (10);
    
    \draw[->] (-0.25,-3) -- (22.5,-3);
    \node at (0,-3)    (t0) {$|$};
    \node at (0,-3.5) {$t$};
    \node at (3.1,-3)    (t1) {$|$};
    \node at (7.3,-3)  (t2a) {$|$};
    \node at (9.8,-3)  (t2) {$|$};
    \node at (16.6,-3) (t3) {$|$};
    \node at (22,-3)   (t5) {$|$};
    \node at (22,-3.5) {$t+1$};

    \draw[-] (t0) -- (t1) node[midway,above] {Protocol phase};
    \draw[-] (t1) -- (t2a) node[midway,above] {Adversary phase};
    \draw[-] (t2a) -- (t2) node[midway,above] {Labeling phase};
    \draw[-] (t2) -- (t3) node[midway,above] {Filtering phase};
    \draw[-] (t3) -- (t5) node[midway,above] {Updating phase};
\end{tikzpicture}
        }
        \caption{Details of round~$t$\textonehalf{} of a $\transitionExt{}{}$-transitional run~$r$.}
        \label{fig:trans_rel}
    \end{center}
\end{figure*} 

\begin{definition}
    The \textbf{causal event filter} returns the set of all attempted events that are ``causally'' possible.
	For a set $X_\epsilon \subseteq \gevents$, sets $X_i \subseteq \gtrueactions[i]$ for each agent $i\in\agents$, and a global history $h = (h_\epsilon, h_1,\dots, h_n) \in \globalstates$, we define 
    \begin{multline}
    \label{eq:filter_function}
    \filtere[B]{h}{X_\epsilon, X_1, \dots, X_n} 
    \ce 
    X_\epsilon
    \setminus 
    \Bigl\{
    \grecv{j}{i}{\mu}{id} 
    \mid 
    \gsend{i}{j}{\mu}{id} \notin h_\epsilon \quad
    \land \\
    (\forall A \in \{\tick\}\sqcup\gactions[i])\ \fakeof{i}{\mistakefor{\gsend{i}{j}{\mu}{id}}{A}} \notin h_\epsilon \quad\land
    \\
    (\gsend{i}{j}{\mu}{id} \notin X_i \lor go(i) \notin X_\epsilon) \quad \land 
    \\ 
    (\forall A \in \{\tick\}\sqcup\gactions[i])\ \fakeof{i}{\mistakefor{\gsend{i}{j}{\mu}{id}}{A}} \notin X_\epsilon
    \Bigr\} 
    \end{multline}
\end{definition}

\begin{definition}
\label{def:ag_history}
    A \textbf{history $h_i$ of agent $i \in \agents$}, or its \textbf{local state}, is a non-empty sequence
    $
    h_i=[\lambda_m,\dots,\lambda_1,\lambda_0]
    $
for some $m\geq0$ such that $\lambda_0\in\Sigma_i$ and $\forall k\in\llbracket1;m\rrbracket$ we have $\lambda_k\subseteq\truethings[i]$. 
In this case $m$ is called the \textbf{length of history} $h_i$ and denoted $|h_i|$. 
We say that a set  $\lambda \subseteq \truethings[i]$ \textbf{is recorded} in the history $h_i$ of agent $i$ and write $\lambda \subseteq h_i$ if{f} $\lambda=\lambda_k$ for some $k\in\llbracket1;m\rrbracket$.
We say that $o \in\truethings[i]$ \textbf{is recorded} in the history $h_i$ and write $o \in h_i$ if{f} $o \in \lambda$ for some set $\lambda\subseteq h_i$. \looseness=-1
 \end{definition}

\begin{definition}
\label{def:env_history}
    A \textbf{history $h$ of the system} with $n$ agents, or the \textbf{global state}, is a tuple
    $
    h\ce (h_\epsilon, h_1, \dots, h_n)
    $
where the \textbf{history of the environment} is a  sequence
$
h_\epsilon = [\Lambda_m, \dots, \Lambda_1]
$
for some $m\geq0$ such that  $\forall k\in\llbracket1;m\rrbracket$ we have $\Lambda_k \subseteq \ghaps$ and $h_i$ is a local state of each agent $i\in \llbracket1;n\rrbracket$. 
In this case $m$ is called the \textbf{length of history}~$h$ and denoted~$|h|\ce|h_\epsilon|$, i.e., the environment has the true global clock. 
We say that a set  $\Lambda \subseteq \ghaps$ \textbf{happens} in the environment's history $h_\epsilon$ or in the system history $h$ and write $\Lambda \subseteq h_\epsilon$ if{f} $\Lambda=\Lambda_k$ for some $k\in\llbracket1;m\rrbracket$.
We say that $O \in \ghaps$ \textbf{happens} in the environment's history~$h_\epsilon$ or in the system history $h$ and write  $O \in h_\epsilon$ if{f} $O \in \Lambda$ for some set $\Lambda \subseteq h_\epsilon$. 
\looseness=-1
\end{definition}

\begin{definition}[Localization function] \label{def:sig_loc_fun}
    The function 
    $ 
    \sigmaof{} \colon 2^{\ghaps}  \longrightarrow  2^{\haps}
    $ 
    is defined as follows
\begin{equation*} 
	\begin{gathered}
		\sigmaof{X} \ce local\Bigl(\mathstrut^{\mathstrut}\bigl(X \cap \gtruehaps \bigr) \ \cup \ \{E \mid (\exists i)\,\fakeof{i}{E}\in X\} \ \cup \\
		\{A' \ne \tick\mid (\exists i)(\exists A)\,\fakeof{i}{\mistakefor{A}{A'}}\in X\} \Bigr).
	\end{gathered}
\end{equation*}
\end{definition}

\begin{definition}[State update functions]\label{def:state-update}
        \index{${\updatee{}{},\updateag{i}{}{}{}}$}
           Given 
           $h = (h_\epsilon, h_1, \dots, h_n) \in \globalstates$, a tuple of performed actions/events $X=(X_{\epsilon},X_1,\dots,X_n) \in 2^{\gevents} \times 2^{\gtrueactions[1]}\times \ldots \times 2^{\gtrueactions[n]}$,
we use the following abbreviation $X_{\epsilon_i} = X_{\epsilon} \cap \gevents[i]$ for 
each 
$i \in \agents$. 
               Agents $i$'s update function 
               \[
               \updateag{i}{}{}{} \colon \localstates{i} \times 2^{\gtrueactions[i]}\times2^{\gevents} \to \localstates{i} 
               \]  
               outputs a new local history from $\localstates{i}$ based on $i$'s actions~$X_i$ and environment-controlled events~$X_{\epsilon}$ as follows:
               \begin{multline}
               \label{eq:update_agent}
                   \updateag{i}{h_i}{X_i}{X_\epsilon} \ce \\
                   \begin{cases}
                   h_i & 
                       \text{if $\sigma(X_{\epsilon_i})=\varnothing$ and } \\
						&X_{\epsilon_i} \cap \sevents[i] \notin \{\{go(i)\}, \{\sleep{i}\}\} \\
                       \Bigl[\sigmaof{ X_{\epsilon_i}\sqcup X_i} \Bigr] : h_i & \text{otherwise }
                   \end{cases}
               \end{multline}
     where $\colon$ represents sequence concatenation.
        Similarly, the environment's state update function 
        $
        \updatee{}{} \colon \localstates{\epsilon} \times \left(2^{\gevents} \times 2^{\gtrueactions[1]}\times \ldots \times 2^{\gtrueactions[n]}\right) \to \localstates{\epsilon}
        $
         outputs a new state of the environment based on $X$:      
        \begin{equation*}
        \updatee{h_\epsilon}{X} \ce (X_{\epsilon} \sqcup X_1 \sqcup \ldots \sqcup X_n) \colon h_\epsilon 
        \end{equation*} 
        Thus, the global state is modified as follows:
        \begin{equation*}
        \label{eq:update_global}
\adjustbox{width=\textwidth}{  $      \update{h}{X} \ce \bigl( \updatee{h_\epsilon}{X}, \updateag{1}{h_1}{X_1}{X_\epsilon}, \dots, \updateag{n}{h_n}{X_n}{X_\epsilon} \bigr)$}
        \end{equation*}
\end{definition}

\begin{definition}
\label{def:t-coherence}
	Let $t \in \bbbt$ be a timestamp.
	A set $S \subset \gevents$ of events is called \textbf{$t$-coherent} if it satisfies the following conditions:
	\begin{enumerate}
		\item\label{coh:fake_send}  for any $\fakeof{i}{\mistakefor{\gsend{i}{j}{\mu}{id}}{A}}\in S$, the GMI $id = id(i,j,\mu,k,t)$ for some $k\in \mathbb{N}$;
		\item\label{coh:one_go} for any $i \in \agents$ at most one event from $\sevents[i]$ is present in $S$;
		\item\label{coh:ext} for any $i \in \agents$ and any locally observable event $e$ at most one of $\glob{i,t}{e}$ and $\fakeof{i}{\glob{i,t}{e}}$ is present in $S$;
		\item\label{coh:recv} for any $\grecv{i}{j}{\mu}{id_1} \in S$, no event of the form $\fakeof{i}{\grecv{i}{j}{\mu}{id_2}}$ belongs to $S$ for any $id_2 \in \mathbb{N}$;
		\item\label{coh:fake_recv} for any $\fakeof{i}{\grecv{i}{j}{\mu}{id_1}} \in S$, no event of the form $\grecv{i}{j}{\mu}{id_2}$ belongs to $S$ for any $id_2 \in \mathbb{N}$;
	\end{enumerate}
\end{definition}



\begin{definition}
    \label{def:consistency}
    \label{def:system}
     For a context $\gamma=(\envprotocol{},\globalinitialstates,\transitionfrom{}{},\Psi)$ and a joint protocol $P$, we define the set of runs \textbf{weakly consistent} with $P$ in $\gamma$ (or weakly consistent with $\chi =(\gamma,\joinprotocol{} )$), denoted $\system{w\chi}=\system{w(\gamma,\joinprotocol{})}$, to be the  set of $\transition{}$-transitional runs that start at  some global initial state from $\globalinitialstates$:
        \begin{equation*}
            \system{w(\gamma,\joinprotocol{})} \ \ce \ \left\{ r \in R \ \mid \ 
                \run{}{0}\in\globalinitialstates \text{\quad and\quad}
                (\forall t\in\bbbt)\, \run{}{t+1}\in\transition{\run{}{t}}
            \right\} 
        \end{equation*}
        
        A run $r$  is called \textbf{strongly consistent}, or simply \textbf{consistent}, with $P$ in~$\gamma$ (or with $\chi$) if it is weakly consistent with $P$ in~$\gamma$ and, additionally, satisfies the admissibility condition:
        $r\in\Psi$. We denote the system of all runs consistent with~$\joinprotocol{}$ in~$\gamma$ by \looseness=-1
 $           \system{(\gamma,\joinprotocol{})} \ce \system{w(\gamma,\joinprotocol{})} \cap \Psi$.
\end{definition}

\begin{definition}
    \index{non-excluding}
    \label{def:nonexcluding}
    An agent-context $\chi$ is \textbf{non-excluding} if{f}
    \[
        \System[\chi]\neq\varnothing \quad\text{ and }\quad  (\forall r\in\System[w\chi]) (\forall t\in\bbbt) (\exists r'\in\System[\chi]) (\forall t'\le t)\, \run[']{}{t'}=\run{}{t'}
	\]
\end{definition}

\begin{definition}
    \index{$\pwrelation{i}$}
    \label{def:possible_world_relation}
    For agent $i\in\agents=\llbracket1;n\rrbracket$, the \textbf{indistinguishability relation} $\pwrelation{i} \subseteq \globalstates^2$ is formally defined by 
   $         \pwrelation{i} \ce \left\{ 
                (h,h') \mid \pi_{i+1}h = \pi_{i+1}{h'}
            \right\}$.
\end{definition}

\begin{definition}
    \index{$\K{i}{\varphi}$}\index{$\eventually{\varphi}$}
    \index{$\everyoneK{G}{\varphi}$}\index{$\commonK{G}{\varphi}$}
    \label{def:semantics}
    Given an interpreted system $\I = (\system{\chi}, \interpretation{}{})$ 
    an agent $i\in\agents$, a run $r\in \system{\chi}$, and a timestamp $t \in \bbbt$:\looseness=-1
    \begin{align*}
        &(\intsys,r,t) \models p &\text{ if{f} }\qquad& r(t) \in \pi(p) \\
        &(\intsys,r,t) \models \neg \varphi &\text{ if{f} }\qquad& (\intsys,r,t) \not\models \varphi\\
        &(\intsys,r,t) \models \varphi \wedge \varphi' &\text{ if{f} }\qquad&(\intsys,r,t) \models \varphi \text{ and }(\intsys,r,t) \models \varphi'\\
        &(\intsys,r,t) \models \K{i}{\varphi} &\text{ if{f} }\qquad& (\forall r'\in R')(\forall t'\in\bbbt) \left(r'_i(t')=r_i(t) \,\Rightarrow\, (\intsys,r',t') \models \varphi\right)
    \end{align*}
\end{definition}

\begin{definition}[Compatibility]
	\label{def:ext_compat}
	For a number of $l \ge 2$ extensions $\extension[\alpha_1]$, $\extension[\alpha_2]$, ..., $\extension[\alpha_l]$ we say the extensions $\extension[\alpha_1]$, $\extension[\alpha_2]$, ..., $\extension[\alpha_l]$ are compatible w.r.t. to some series of extension combinations $\star_1$, $\star_2$, ..., $\star_{l-1}$ \footnote{For some $l' \in \mathbb{N}$ we use $\alpha \star_{l'} \beta$ or just $\star$ to represent either forth composition ($\alpha \circ \beta$) or reversed composition ($\beta \circ \alpha$). Note that in our complete framework we distinguish between further types of combinations.} if{f} 
	$\mathit{PP}^\alpha_1 \cap \ldots \cap \mathit{PP}^\alpha_l \ne \varnothing$, 
	$\mathit{IS}^\alpha_1 \cap \ldots \cap \mathit{IS}^\alpha_l \ne \varnothing$,
	$\Admissibility[\alpha_1] \cap \ldots \cap \Admissibility[\alpha_l]\ne\varnothing$ and
	$\exists \chi \in \extension[\alpha_1 \star_1 \alpha_2 \star_2 \ldots \star_{l-1} \alpha_l]$.

    If{f} extensions $\extension[\alpha_1]$, $\extension[\alpha_2]$, ..., $\extension[\alpha_l]$ ($l \ge 2$) are \textbf{compatible} w.r.t. the extension combination series $\star_1$, $\star_2$, ..., $\star_{l-1}$, then $\extension[\alpha_1 \star_1 \alpha_2 \star_2 \ldots \star_{l-1} \alpha_l]$ is also an extension.
\end{definition}

\begin{definition} \label{def:t-coherent_domain}
	We define $\mathit{PD}_\epsilon^{t\textup{-coh}}$ as \emph{the (downward closed) domain of all $t$-coherent events}:
		$\mathit{PD}_\epsilon^{t\textup{-coh}} \ce \{X_\epsilon \in 2^{\gevents} \mid X_\epsilon \text{ is $t$-coherent for some } t \in \bbbt\}$.
\end{definition}

\begin{definition}
\label{def:liveness_prop}
	We define a \textbf{liveness} property as a subset  $\liv \subseteq \system{}$, where 
	$\liv \ne \varnothing \ \wedge \ (\forall r \in \systrans)(\forall t \in \bbbt)(\exists r' \in \liv)\ r'(t) = r(t)$.
\looseness=-1
\end{definition}
Informally, liveness says that every prefix $r(t)$ of every  run $r$  can be extended in~$L$.\looseness=-1

\begin{definition}
\label{def:adherence_lifeness_prop}
	An extension $\extension[\alpha]$ \textbf{adheres to} a safety property $S'$ (resp. liveness property $\liv$) if{f}
		$\ \bigcup_{\chi^\alpha \in \extension[\alpha]} \system{\chi^\alpha} \ \subseteq \ S'$ 
		$(\bigcup_{\chi^\alpha \in \extension[\alpha]} \system{\chi^\alpha} \ \subseteq \ \liv)$.
\end{definition}

\begin{definition} \label{def:saf_liv_prep}
	For a set $P^\alpha \subseteq \systrans$ of transitional runs, where $P^\alpha \ne \varnothing$, 
	\begin{align}
		L'^\alpha &\ce \{r \in \systrans \mid (\exists t \in \bbbt)(\forall r' \in P^\alpha)(\forall t' \in \bbbt)\ r(t) \ne r'(t')\} \\
		\overline{L^\alpha} &\ce P^\alpha \cup L'^\alpha. \label{eq:saf_liv_prep}
	\end{align}
\end{definition}

\begin{lemma} \label{lem:L_is_live}
	$\overline{L^\alpha}$ is a liveness property.
\end{lemma}
\begin{proof}
	Since $P^\alpha \ne \varnothing$,  $\overline{L^\alpha} \ne \varnothing$ as well by (\ref{eq:saf_liv_prep}).

	Take any finite prefix $r(t)$ of a run $r \in \systrans$ for some timestamp $t \in \bbbt$.
	If $r(t)$ has an extension in $P^\alpha$, then there exists a run $r' \in P^\alpha$, s.t. $r(t) = r'(t)$.
	Since by Def.~\ref{def:saf_liv_prep} $P^\alpha \subseteq \overline{L^\alpha}$, $r' \in \overline{L^\alpha}$ as well.
	If $r(t)$ has no extension in $P^\alpha$, then by Def.~\ref{def:saf_liv_prep} $r \in L'^\alpha$, thus $r \in \overline{L^\alpha}$.
\end{proof}

\begin{definition} \label{def:trace_saf}
	The smallest trace safety property containing $P \subseteq \systrans$, for $P \ne \varnothing$, is the \textbf{prefix and limit closure} of $P$, formally
	\begin{align*}
			S'(P) \ce \{h \in \prefset \mid &(\exists r \in P)(\exists t \in \bbbt)\ r(t) = h\} \sqcup \\
				&\{r \in \systrans \mid (\forall t \in \bbbt)(\exists r' \in P)\ r(t) = r'(t) \}.
	\end{align*}
	The set of all trace safety properties is denoted by $\tracesaf$.
\end{definition}

\begin{lemma} \label{lem:P_is_L_cap_S}
	$P^\alpha = \overline{L^\alpha} \cap S'(P^\alpha)$, where $S'(P^\alpha)$ is the prefix and limit closure of $P^\alpha$ (see Def.~\ref{def:trace_saf}).
\end{lemma}
\begin{proof}
	Since $P^\alpha \subseteq S'(P^\alpha)$ and $P^\alpha \subseteq \overline{L^\alpha}$, it follows that $P^\alpha \subseteq \overline{L^\alpha} \cap S'(P^\alpha)$.
	Hence it remains to show that $\overline{L^\alpha} \cap S'(P^\alpha) \subseteq P^\alpha$.
	Assume by contradiction that there exists a run $r \in \overline{L^\alpha} \cap S'(P^\alpha)$, but $r \notin P^\alpha$, hence $r \in \overline{L^\alpha}$---specifically $r \in L'^\alpha$---and $r \in S'(P^\alpha)$.
	Since $r \in S'(P^\alpha)$ (by prefix closure of $S'(P^\alpha)$) for all $t' \in \bbbt$, $r(t') \in S'(P^\alpha)$ as well.
	This implies (by limit closure of $S'(P^\alpha)$) that for all $t \in \bbbt$ there must exist a run $r' \in P^\alpha$ such that $r(t) = r'(t)$. 
	This however contradicts that $r \in L'^\alpha$.
\end{proof}

\begin{definition} \label{def:bij_constr_tr_to_op}
	A construction $F'$ of an operational safety property from a trace safety property $S' \in \tracesaf$ is 
$		F(S')(h) \ce \{\betaag{}{r}{t} \mid r \in S' \ \wedge \ t \in \bbbt \ \wedge \ h = r(t)\}$.
\end{definition}

\begin{lemma} \label{lem:F_mapping}
	$F'(S') \in \opsaf$ for any $S' \in \tracesaf$.
\end{lemma}
\begin{proof}
	Suppose by contradiction there exists some $S'^\alpha \in \tracesaf$ s.t. $F'(S'^\alpha) = S^\alpha$, where $S^\alpha$ violates the first operational safety property attribute (\ref{item:op_saf_attr_1}).
	This implies
$		(\forall h \in \prefset)\ h_\epsilon \ne [] \ \vee \ S^\alpha(h) = \varnothing$.
	Since by Def.~\ref{def:trace_saf} $P^\alpha \ne \varnothing$, we get that there has to exist a run $r \in S'^\alpha$.
	Further, by prefix closure of $S'^\alpha$, we have
$		(\forall t \in \bbbt)\ r(t) \in S'^\alpha$,
	from which by universal instantiation we get that $r(0) \in S'^\alpha$.
	Since $S'^\alpha \subseteq \systrans \sqcup \prefset$, $r$ is transitional, hence $r_\epsilon(0) = []$, from which by our assumption $S^\alpha(r(0)) = \varnothing$ follows.
	However, by Def.~\ref{def:bij_constr_tr_to_op} of construction $F'$, it follows that $\betaag{}{r}{0} \in S^\alpha(r(0))$, thus $S^\alpha(r(0)) \ne \varnothing$.

	Next, suppose by contradiction there exists some $S'^\alpha \in \tracesaf$ s.t. $F'(S'^\alpha) = S^\alpha$, where $S^\alpha$ violates the second operational safety property attribute (\ref{item:op_saf_attr_2}).
	This implies that there exists some $h \in \prefset$ s.t. $h_\epsilon \ne []$ and
	\begin{gather}
	\label{eq:lem_F_mapping_attr2_contra1} 
		(((\exists h' \in \prefset)(\exists X \in S^\alpha(h')) h = \update{h'}{X})  \wedge S^\alpha(h) = \varnothing)  \vee \\
		 \label{eq:lem_F_mapping_attr2_contra2}
		(((\forall h'' \in \prefset)(\forall X' \in S^\alpha(h'')) h \ne \update{h''}{X'})  \wedge  S^\alpha(h) \ne \varnothing).
	\end{gather}
	Suppose (\ref{eq:lem_F_mapping_attr2_contra1}) is true.
	This implies that there exists some $h' \in \prefset$ and some $X \in S^\alpha(h')$ such that $h = \update{h'}{X}$.
	By Def.~\ref{def:bij_constr_tr_to_op} of $F'$ there exists a run $r \in S'^\alpha$ and a timestamp $t \in \bbbt$ s.t. 
$		r(t) = h'$ and  
$X = \betaag{}{r}{t}$.
	By transitionality of $r$ and Def.~\ref{def:state-update} of $\update{}{}$
$		r(t+1) = h$.
	Again by Def.~\ref{def:bij_constr_tr_to_op}, we have
$		\betaag{}{r}{t+1} \in S^\alpha(h)$,
	hence $S^\alpha(h) \ne \varnothing$ and we conclude that (\ref{eq:lem_F_mapping_attr2_contra1}) is false.

	Suppose (\ref{eq:lem_F_mapping_attr2_contra2}) is true.
	This implies by Def.~\ref{def:bij_constr_tr_to_op} that there exists a run $r \in S'^\alpha$ and timestamp $t \in \bbbt \setminus \{0\}$, where $h = r(t)$, since $r$ is transitional and $h_\epsilon \ne []$.
	Further we get that
$		\betaag{}{r}{t-1} \in S^\alpha(r(t-1))$.
	Thus by Def.~\ref{def:state-update} of $\update{}{}$, we have
$		r(t) = \update{r(t-1)}{\betaag{}{r}{t-1}}$
	and we conclude that (\ref{eq:lem_F_mapping_attr2_contra2}) is false as well.\looseness=-1
\end{proof}

\begin{definition} \label{def:bij_map_tr_to_op}
	We define
$		F \colon \tracesaf \mapsto \opsaf$,
	where for any $S' \in \tracesaf$ we have
$		F(S') \ce F'(S')$,
	for $F'$ from Def.~\ref{def:bij_constr_tr_to_op}, which is indeed a mapping from $\tracesaf$ to $\opsaf$ by Lemma \ref{lem:F_mapping}.
\end{definition}

\begin{lemma} \label{lem:f_is_injective}
	$F$ from Def.~\ref{def:bij_map_tr_to_op} is injective.
\end{lemma}
\begin{proof}
	Suppose by contradiction that the opposite is true: there are $S'^\alpha, S'^\beta \in \tracesaf$ s.t. $S'^\alpha \ne S'^\beta$, but $F(S'^\alpha) = F(S'^\beta)$.
	Since $S'^\alpha \ne S'^\beta$, either
	\begin{compactenum} [1]
	\item w.l.o.g. there exists some history $h \in S'^\alpha$ s.t. $h \notin S'^\beta$ or \label{item:f_injective}
	\item w.l.o.g. there exists some run $r \in S'^\alpha$ s.t. $r \notin S'^\beta$. We show that this implies~\ref{item:f_injective}.
		Suppose by contradiction that there does not exist some $h \in S'^\alpha$ s.t. $h \notin S'^\beta$, meaning $(\forall h \in S'^\alpha)\ h \in S'^\beta$.
		By limit closure of $S'^\beta$ however it follows that $r \in S'^\beta$, hence there has to exist a history $h \in S'^\alpha$ such that $h \notin S'^\beta$.\looseness=-1
	\end{compactenum}
	Therefore, we can safely assume \ref{item:f_injective}, i.e., w.l.o.g. that there exists some $h \in S'^\alpha$ s.t. $h \notin S'^\beta$.
	By Def.~\ref{def:bij_map_tr_to_op} of $F$, we get that $F(S'^\beta)(h) = \varnothing$, as otherwise there would exist a run $r' \in S'^\beta$ and time $t' \in \bbbt$ s.t. $r'(t') = h$, from which by prefix closure of $S'^\beta$ it would follow that $h \in S'^\beta$.
	Since $S'^\alpha$ is the prefix closure of some non-empty set $P^\alpha \subseteq \systrans$ by Def.~\ref{def:trace_saf}, we get that there exists some run $r \in S'^\alpha$ and time $t \in \bbbt$ s.t. $r(t) = h$, additionally by Def.~\ref{def:bij_map_tr_to_op} of $F$, $\betaag{}{r}{t} \in F(S'^\alpha)(h)$.
	Therefore $F(S'^\alpha) \ne F(S'^\beta)$ and we are done.
\end{proof}

\begin{definition} \label{def:f_surj_construction}
	For some arbitrary $S \in \opsaf$, we define
	\begin{align}
		\widetilde{S'_0}^{S} &\ce \systrans \label{eq:lem_f_surjective_constr_init} \\
		\widetilde{S'_t}^{S} &\ce \widetilde{S'_{t-1}}^{S} \setminus \{r \in \systrans \mid \betaag{}{r}{t-1} \notin S(r(t-1))\} \label{eq:lem_f_surjective_constr_induc} \\
		\widetilde{S'_{\infty}}^{S} &\ce \lim\limits_{t' \rightarrow \infty} \widetilde{S'_{t'}}^{S} \label{eq:lem_f_surjective_constr_limit} \\
		\widetilde{S'}^{S} &\ce \widetilde{S'_{\infty}}^{S} \ \sqcup \ \{h \in \prefset \mid (\exists r \in \widetilde{S'_{\infty}}^{S})(\exists t \in \bbbt)\ h = r(t)\}. \label{eq:lem_f_surjective_constr_prefix_cl}
	\end{align}
	Note that the limit in (\ref{eq:lem_f_surjective_constr_limit}) exists, as by (\ref{eq:lem_f_surjective_constr_induc}) the set $\widetilde{S'_t}^{S}$ is non-increasing in $t$.\looseness=-1
\end{definition}

\begin{lemma} \label{lem:f_surj_construction_is_limit_closure_thing}
	For $\widetilde{S'_m}^{\widetilde{S}}$ (for $m \in \bbbt \setminus \{0\}$ and $\widetilde{S} \in \opsaf$) from Def.~\ref{def:f_surj_construction}, it holds that \looseness=-1
$		\widetilde{S'_m}^{\widetilde{S}} = \{r \in \systrans \mid (\forall t < m)\ \betaag{}{r}{t} \in \widetilde{S}(r(t))\}$.
\end{lemma}
\begin{proof}
	By induction:
	\\
	\textbf{Induction Hypothesis:}
	\begin{equation} \label{eq:lem_f_surf_constr_limit_cl_ind_hyp}
		\widetilde{S'_m}^{\widetilde{S}} = \{r \in \systrans \mid (\forall t < m)\ \betaag{}{r}{t} \in \widetilde{S}(r(t))\}.
	\end{equation}
	\\
	\textbf{Base Case:} For $m = 1$ by Def.~\ref{def:f_surj_construction} it follows that
	\begin{equation*}
		\widetilde{S'_1}^{\widetilde{S}} = \systrans \setminus \{r \in \systrans \mid \betaag{}{r}{0} \notin \widetilde{S}(r(0))\} = \{r \in \systrans \mid \betaag{}{r}{0} \in \widetilde{S}(r(0))\}.
	\end{equation*}
	
	\noindent
	\textbf{Induction Step:} Suppose the induction hypothesis (\ref{eq:lem_f_surf_constr_limit_cl_ind_hyp}) holds for $m$, but by contradiction does not hold for $m+1$.
	There are two cases:
	\begin{enumerate}
	\item There exists a run $r' \in \widetilde{S'_{m+1}}^{\widetilde{S}}$ s.t. $r' \notin \{r \in \systrans \mid (\forall t < m+1)\ \betaag{}{r}{t} \in \widetilde{S}(r(t))\}$.
		This implies that there exists some timestamp $t' < m+1$ s.t.
$			\betaag{}{r'}{t'} \notin \widetilde{S}(r'(t'))$.
		We distinguish two cases:
		\begin{enumerate}
		\item $t' = m$: Then
	$			r' \in \{r \in \systrans \mid \betaag{}{r}{m} \notin \widetilde{S}(r(m))\}$.
			Hence by Def.~\ref{def:f_surj_construction} of $\widetilde{S'_{m+1}}^{\widetilde{S}}$, $r'$ would have been removed.
		\item $t' < m$: This directly contradicts the induction hypothesis (\ref{eq:lem_f_surf_constr_limit_cl_ind_hyp}), as $r' \notin \widetilde{S'_{t'+1}}^{\widetilde{S}}$ and $\widetilde{S'_{m+1}}^{\widetilde{S}} \subseteq \widetilde{S'_{t'+1}}^{\widetilde{S}}$.
		\end{enumerate}
	\item There exists a run $r' \in \{r \in \systrans \mid (\forall t < m+1)\ \betaag{}{r}{t} \in \widetilde{S}(r(t))\}$ s.t. $r' \notin \widetilde{S'_{m+1}}^{\widetilde{S}}$.
		We distinguish two cases regarding at which step $r$ has been removed:
		\begin{enumerate}
		\item $r' \in \widetilde{S'_m}^{\widetilde{S}}$: Then
	$			r' \in \{r \in \systrans \mid \betaag{}{r}{m} \notin \widetilde{S}(r(m))\}$.
			This implies that $\betaag{}{r'}{m} \notin \widetilde{S}(r'(m))$ contradicting 
			$	r' \in \{r \in \systrans \mid (\forall t < m+1)\ \betaag{}{r}{t} \in \widetilde{S}(r(t))\}$.
		\item $r' \notin \widetilde{S'_m}^{\widetilde{S}}$: This directly contradicts the induction hypothesis (\ref{eq:lem_f_surf_constr_limit_cl_ind_hyp}), thus concluding the induction step.
		\end{enumerate}
	\end{enumerate}
\end{proof}

\begin{lemma} \label{lem:f_surf_constr_limit_is_limit_cl_thing}
	For $\widetilde{S'_{\infty}}^{S}$ from Def.~\ref{def:f_surj_construction} it holds that
	\begin{equation*} 
		\widetilde{S'_{\infty}}^{S} = \{r \in \systrans \mid (\forall t \in \bbbt)\ \betaag{}{r}{t} \in S(r(t))\}
	\end{equation*}
\end{lemma}
\begin{proof}
	Follows from Lemma \ref{lem:f_surj_construction_is_limit_closure_thing} and Def.~\ref{def:f_surj_construction}.
\end{proof}

\begin{lemma} \label{lem:f_surf_constr_is_trace_saf_prop}
	For $\widetilde{S'}^{S}$, from Def.~\ref{def:f_surj_construction}, where $S \in \opsaf$, it holds that $\widetilde{S'}^{S} \in \tracesaf$, i.e. $\widetilde{S'}^{S}$ is a trace safety property.
\end{lemma}
\begin{proof}
	From Def.~\ref{def:f_surj_construction}, particularly (\ref{eq:lem_f_surjective_constr_limit}) and (\ref{eq:lem_f_surjective_constr_prefix_cl}), it follows that $\widetilde{S'}^{S}$ is the prefix and limit closure of $\widetilde{S'_{\infty}}^{S}$.
\end{proof}

\begin{lemma} \label{lem:update_injective}
	The state update function---$\update{}{}$---from Def.~\ref{def:state-update} is injective.
\end{lemma}
\begin{proof}
	Recall that according to Def.~\ref{def:state-update} $\update{}{}$ and its constituent parts are defined for $h \in \globalstates$, $i \in \agents$ and $X \in 2^{\gevents} \times 2^{\gtrueactions[1]} \times \ldots \times 2^{\gtrueactions[n]}$ as
	\begin{gather}
\adjustbox{width=.88\textwidth}{	$	\update{h}{X} \ce \bigl( \updatee{h_\epsilon}{X}, \updateag{1}{h_1}{X_1}{X_\epsilon}, \dots, \updateag{n}{h_n}{X_n}{X_\epsilon} \bigr)$} \label{eq:update_injective} \\
\adjustbox{width=.88\textwidth}{	$		\updateag{i}{h_i}{X_i}{X_\epsilon} \ce \label{eq:updateag_injective}
		\begin{cases}
			h_i & \text{if $\sigma(X_{\epsilon_i})=\varnothing$ and } \unaware{i}{X_\epsilon}  \\
			\bigl[\sigmaof{ X_{\epsilon_i}\sqcup X_i} \bigr] : h_i & \text{otherwise }
		\end{cases} $}\\
		\updatee{h_\epsilon}{X} \ce (X_{\epsilon} \sqcup X_1 \sqcup \ldots \sqcup X_n) \colon h_\epsilon. \label{updatee_injective}
	\end{gather}
	Suppose by contradiction that $\update{}{}$ is not injective, i.e., $\update{h^1}{X^1} = \update{h^2}{X^2}$ for some  $(h^1,X^1)\ne (h^2,X^2) \in \globalstates \times 2^{\gevents} \times 2^{\gtrueactions[1]} \times \ldots \times 2^{\gtrueactions[n]}$.
	We distinguish the following cases:
	\begin{compactenum}
		\item $X^1 \ne X^2$: By (\ref{updatee_injective}) $X^1 : h^1_\epsilon \ \ne \ X^2 : h^2_\epsilon$, as irrespective of $h^1$ and $h^2$ the two resulting histories have different suffixes (of size one) $X^1$ and $X^2$ (recall that by Defs.~\ref{def:ag_history} and~\ref{def:env_history} the environment and the agent histories are sequences of sets). \label{item:upate_inj_diff_X}
		\item $h^1_\epsilon \ne h^2_\epsilon$:
		If $X_1 \ne X_2$ it follows from case (\ref{item:upate_inj_diff_X}) that $X^1 : h^1_\epsilon \ne X^2 : h^2_\epsilon$. \\
		Else if $X_1 = X_2$: if further $|h^1_\epsilon| = |h^2_\epsilon|$ then the two resulting histories now have the same suffix $X^1$, however still different prefixes $h^1_\epsilon$ and $h^2_\epsilon$. \\
		Otherwise if w.l.o.g. $|h^1_\epsilon| > |h^2_\epsilon|$, then $|X^1 : h^1_\epsilon| > |X^1 : h^2_\epsilon|$ and we are done.
		\item $h^1_i \ne h^2_i$ (for some $i \in \agents$): if
		\begin{compactitem}
			\item $X^1 \ne X^2$: follows from case (\ref{item:upate_inj_diff_X})
			\item $X^1 = X^2$: either
$				\updateag{i}{h^1_i}{X^1_i}{X^1_\epsilon} = h^1_i$ 
or 
	$			\updateag{i}{h^2_i}{X^1_i}{X^1_\epsilon} = h^2_i$,
			since $h^1_i \ne h^2_i$ we conclude that the resulting (local) histories are different, or
			\begin{align*}
				\updateag{i}{h^1_i}{X^1_i}{X^1_\epsilon} &= \bigl[\sigmaof{ X^1_{\epsilon_i}\sqcup X^1_i} \bigr] : h^1_i \\
				\updateag{i}{h^2_i}{X^1_i}{X^1_\epsilon} &= \bigl[\sigmaof{ X^1_{\epsilon_i}\sqcup X^1_i} \bigr] : h^2_i.
			\end{align*}
			Suppose that $|h^1_i| = |h^2_i|$.
			It follows that the two resulting (local) histories have matching suffixes ($\Bigl[\sigmaof{ X^1_{\epsilon_i}\sqcup X^1_i} \Bigr]$), however different prefixes.
			If on the other hand w.l.o.g. $|h^1_i| > |h^2_i|$, it also follows that $\big|\bigl[\sigmaof{ X^1_{\epsilon_i}\sqcup X^1_i} \bigr] : h^1_i\big| > \big|\bigl[\sigmaof{ X^1_{\epsilon_i}\sqcup X^1_i} \bigr] : h^2_i\big|$, hence they cannot be the same and we are done.\looseness=-1
		\end{compactitem}
	\end{compactenum}
\end{proof}

\begin{lemma} \label{lem:op_saf_prop_once_empty_always_empty}
	Given an operational safety property $S \in \opsaf$, a transitional run $r \in \systrans$ and timestamps $t,t' \in \bbbt$, where $t' \ge t$, if
$		S(r(t)) = \varnothing$, then 
$S(r(t')) = \varnothing$.
\end{lemma}
\begin{proof}
	By induction over $t' \ge t$. \\
	\textbf{Induction Hypothesis:} For $t' \ge t$ and $S \in \opsaf$ it holds that
if $		S(r(t)) = \varnothing$,
then 
$S(r(t')) = \varnothing$.
\\
	\textbf{Base Case} for $t' = t$: it trivially follows that $S(r(t')) = S(r(t)) = \varnothing$. \\
	\textbf{Induction Step} for $t' \rightarrow t'+1$: 
	Since the state update function is injective by Lemma \ref{lem:update_injective} and run $r$ is transitional, the only way to achieve the prefix $r(t'+1)$ via state update is by $\mathit{update}({r(t')},{\betaag{}{r}{t'}})$.
	However since by the induction hypothesis $S(r(t')) = \varnothing$, it follows by the second operational safety property attribute (\ref{item:op_saf_attr_2}) that $S(r(t'+1)) = \varnothing$ as well.
\end{proof}

\begin{lemma} \label{lem:saf_prefix_impl_saf_run}
	For any $h \in \prefset$ and operational safety property $S \in \opsaf$ it holds that $(\exists r \in \systrans)(r(|h|) = h ) \wedge ((\forall t \in \bbbt)\ S(r(t)) \ne \varnothing $ whenever 
$		S(h) \ne \varnothing$.
\end{lemma}
\begin{proof}
	Assuming that $S(h) \ne \varnothing$ we construct $r$ as follows:
	Since $h \in \prefset$, $h = r'(|h|)$ for some $r' \in \systrans$.
	Contraposition of the statement of Lemma \ref{lem:op_saf_prop_once_empty_always_empty} gives for $t \le |h|$, if 
$		S(r'(|h|)) \ne \varnothing$ 
then $S(r'(t')) \ne \varnothing$.
	Hence for $t \le |h|$ we define $r(t) \ce r'(t)$.

	Next assume an order on the set $Z \ce 2^{\gevents} \times 2^{\gactions[1]} \times \ldots \times 2^{\gactions[n]}$ and let $\widetilde{X_1}(S)$ be the first element of some subset $S \subseteq Z$ according to this order.
	For $t > |h|$ we define $r(t) \ce \mathit{update}({r(t-1)},{\widetilde{X_1}(S(r(t-1)))})$.
	It remains to show that $S(r(t-1))$ can never be empty.
	The proof is by induction over $t > |h|$. \\
	\textbf{Induction Hypothesis:} $S(r(t-1)) \ne \varnothing$. \\
	\textbf{Base Case} for $t = |h| + 1$: We get that $S(t) = S(|h|)$, which is not empty by assumption. \\
	\textbf{Induction Step} for $t \rightarrow t + 1$:
	Suppose the induction hypothesis holds for~$t$.
	Since $r(t)$ is defined as $\mathit{update}({r(t-1)},{\widetilde{X_1}(S(r(t-1)))})$ and it holds that $r(t-1) \in \prefset$ and $\widetilde{X_1}(S(r(t-1))) \in S(r(t-1))$, by the second operational safety property attribute (\ref{item:op_saf_attr_2}) and semantics of $\Longleftrightarrow$ we get that also $S(r(t)) \ne \varnothing$, thus completing the induction step.
\end{proof}

\begin{lemma} \label{lem:saf_not_empty_for_run_impl_beta_set}
	For an operational safety property $S \in \opsaf$, transitional run $r \in \systrans$, and timestamp $t \in \bbbt \setminus \{0\}$, 
 $		S(r(t)) \ne \varnothing$ 
implies $\betaag{}{r}{t-1} \in S(r(t-1))$.
\looseness=-1
\end{lemma}
\begin{proof}
	If $S(r(t)) \ne \varnothing$, by the operational safety property attribute \ref{item:op_saf_attr_2}, $r(t)$ has to be safely reachable, meaning
$		(\exists h \in \prefset)(\exists X \in S(h)) \ r(t) = \update{h}{X}$.
	By injectivity (Lemma \ref{lem:update_injective}) of $\update{}{}$ (Def.~\ref{def:state-update}) it only maps to the prefix $r(t)$ for $h = r(t-1)$ and $X = \betaag{}{r}{t-1}$.
	Hence $\betaag{}{r}{t-1} \in S(r(t-1))$.
\end{proof}

\begin{lemma} \label{lem:f_is_surjective}
	$F$ from Def.~\ref{def:bij_map_tr_to_op} is surjective.
\end{lemma}
\begin{proof}
	Suppose by contradiction that $F$ is not surjective.
	This implies that there exists some $S \in \opsaf$ s.t. for all $S' \in \tracesaf$, $F(S') \ne S$.

	To arrive at a contradiction, we use the trace safety property $\widetilde{S'}^{S}$ from Def.~\ref{def:f_surj_construction}.
	This is safe to use, as by Lemma \ref{lem:f_surf_constr_is_trace_saf_prop} $\widetilde{S'}^{S} \in \tracesaf$.
	There are two cases causing $F(\widetilde{S'}^S) \ne S$:
	\begin{compactenum}
	\item There is a run $r' \in \widetilde{S'}^{S}$ and timestamp $t' \in \bbbt$ s.t. $\betaag{}{r'}{t'} \notin S(r'(t'))$.
		Hence,
$			r' \in \{r \in \systrans \mid \betaag{}{r}{t'} \notin S(r(t'))\}$,
		such that by (\ref{eq:lem_f_surjective_constr_induc}) $r' \notin \widetilde{S'_{t'+1}}^{S}$, from which further by (\ref{eq:lem_f_surjective_constr_limit}) and (\ref{eq:lem_f_surjective_constr_prefix_cl}) $r' \notin \widetilde{S'}^{S}$ follows, providing a contradiction.
	\item There exists a prefix $h \in \prefset$ and some $X \in S(h)$, but
		\begin{equation} \label{eq:lem_f_surjective_contra2}
			h \notin \widetilde{S'}^{S}.
		\end{equation}
		Since $X \in S(h)$ by Lemma \ref{lem:saf_prefix_impl_saf_run} there exists a transitional run $r \in \systrans$ s.t. 
$			r(|h|) = h $ and  $(\forall t \in \bbbt)\ S(r(t)) \ne \varnothing $.
		By Lemma \ref{lem:saf_not_empty_for_run_impl_beta_set} we further get that for any $t \in \bbbt$, if 
$			S(r(t)) \ne \varnothing $, then 
$\betaag{}{r}{t-1} \in S(r(t-1))$.
		By Lemma \ref{lem:f_surf_constr_limit_is_limit_cl_thing} it follows that $r \in \widetilde{S'_{\infty}}^{S}$ and by prefix closure (\ref{eq:lem_f_surjective_constr_prefix_cl}) we finally get that $h \in \widetilde{S'}^{S}$ contradicting (\ref{eq:lem_f_surjective_contra2}).
	\end{compactenum}

	Thus, by definition of our construction (\ref{eq:lem_f_surjective_constr_init})--(\ref{eq:lem_f_surjective_constr_limit}), $F(\widetilde{S'}^{S}) = S$.\looseness=-1
\end{proof}

\begin{lemma} \label{lem:f_is_bijective}
	$F$ from Def.~\ref{def:bij_constr_tr_to_op} is bijective.
\end{lemma}
\begin{proof}
	Follows from Lemma \ref{lem:f_is_injective} and \ref{lem:f_is_surjective}.
\end{proof}

\begin{lemma}\label{lem:synch_gbyz_indist}
For the general asynchronous byzantine framework given two $\tauprotocol{B}{\envprotocol{}}{\joinprotocol{}}$-transitional runs $r,r' \in \System$ and timestamps $t,t' \in \bbbt \setminus \{0\}$, an agent $i \in \agents$ cannot distinguish
	\begin{compactitem}
	\item a round $t$\textonehalf{} in run $r$, where a nonempty set of events $Q \subseteq \gtrueevents[i] \sqcup \bevents[i]$ was observed by $i$, but no $go(i)$ occurred $\Rightarrow$
	\begin{equation*}
		go(i) \notin \betag[i]{r}{t}, \quad \betaag{i}{r}{t} = \varnothing, \quad \betaout[i]{r}{t} \sqcup \betab[i]{r}{t} = Q
	\end{equation*}
	\item from a round $t'$\textonehalf{} in run $r'$, where the same set of events $Q$ was observed by $i$, $go(i)$ occurred, but the protocol prescribed the empty set ($\varnothing \in \agprotocol{i}{r'(t')}$), which was chosen by the adversary $\Rightarrow$
	\begin{equation*}
		go(i) \in \betag[i]{r'}{t'}, \quad \betaag{i}{r'}{t'} = \varnothing, \quad \betaout[i]{r'}{t'} \sqcup \betab[i]{r'}{t'} = Q.
	\end{equation*}
	\end{compactitem}
\end{lemma}
\begin{proof}
This immediately follows from the definition of the update function (\ref{eq:update_agent}), as in this scenario (for $r_i(t+1) = [\lambda_m,\dots,\lambda_1,\lambda_0]$ and $r'_{i}(t'+1) = [\lambda'_{m'},\dots,\lambda'_1,\lambda'_0]$) $\lambda_m = \lambda'_{m'} = Q$. \qed
\end{proof}

\begin{definition} \label{def:impl_classes_extended}
	We define the following implementation classes:
	\begin{itemize}
	\item[$\adm$] 
		The desired extension property is only implemented via an appropriate admissibility condition $\Admissibility[\alpha] \subseteq \system{}$.
		An extension $\extension[\alpha] \in \adm$ if{f} $\extension[\alpha]=(\envprotocols \times \agprotocols, IS^\alpha, \transitionfrom[N,N]{}{}, \Admissibility[\alpha])$.
	\item[$\jp$]
		The extension property is implemented via restricting the set of joint protocols $\agprotocols$.
		An extension $\extension[\alpha] \in \jp$ if{f}
			$\extension[\alpha]=(\envprotocols \times \agprotocols[\alpha], IS^\alpha, \transitionfrom[N,N]{}{}, \Admissibility[\alpha])$.
	\item[$\jpfb$]
		An extension $\extension[\alpha] \in \jpfb$ if{f}
			$\extension[\alpha]=(\envprotocols \times \agprotocols[\alpha], IS^\alpha, \transitionfrom[N,B]{}{}, \Admissibility[\alpha])$,
		where in $\transitionfrom[N,B]{}{}$ the filter functions $\filtere[N]{}{}$ and $\filterag[B]{i}{}{}$ (for all $i \in \agents$) are used and $\agprotocols[\alpha] \subset \agprotocols$.
	\item[$\ejp$]
		The extension property is implemented via restricting the set of environment protocols $\envprotocols$ possibly in conjunction with the set of joint protocols $\agprotocols$.
		An extension $\extension[\alpha] \in \ejp$ if{f}
			$\extension[\alpha]=(PP^{\alpha}, IS^\alpha, \transitionfrom[N,N]{}{}, \Admissibility[\alpha])$,
		where $PP^{\alpha} \subset \pprotocols$ and $\extension[\alpha] \notin \jp$.
	\item[$\ejpfb$]
		An extension $\extension[\alpha] \in \ejpfb$ if{f}
			$\extension[\alpha]=(PP^{\alpha}, IS^\alpha, \transitionfrom[N,B]{}{}, \Admissibility[\alpha])$,
		where in $\transitionfrom[N,B]{}{}$ the filter functions $\filtere[N]{}{}$ and $\filterag[B]{i}{}{}$ (for all $i \in \agents$) are used, $PP^{\alpha} \subset \pprotocols$ and $\extension[\alpha] \notin \jpfb$.
	\item[$\efjp$]
		An extension $\extension[\alpha] \in \efjp$ if{f}
			$\extension[\alpha]=(\envprotocols \times \agprotocols[\alpha], IS^\alpha, \transitionfrom[\alpha,N]{}{}, \Admissibility[\alpha])$,
		where in $\transitionfrom[\alpha,N]{}{}$ the filter functions $\filtere[\alpha]{}{}$ and $\filterag[N]{i}{}{}$ (for all $i \in \agents$) are used, $\agprotocols[\alpha] \subseteq \agprotocols$ and $\extension[\alpha] \notin \jp$.
		\item[$\efjpfb$]
		An extension $\extension[\alpha] \in \efjpfb$ if{f}
			$\extension[\alpha]=(\envprotocols \times \agprotocols[\alpha], IS^\alpha, \transitionfrom[\alpha,B]{}{}, \Admissibility[\alpha])$,
		where in $\transitionfrom[\alpha,B]{}{}$ the filter functions $\filtere[\alpha]{}{}$ and $\filterag[B]{i}{}{}$ (for all $i \in \agents$) are used, $\agprotocols[\alpha] \subseteq \agprotocols$ and $\extension[\alpha] \notin \jpfb$.
		\item[$\efejp$] 
		An extension $\extension[\alpha] \in \efejp$ if{f}
			$\extension[\alpha]=(PP^\alpha, IS^\alpha, \transitionfrom[\alpha,N]{}{}, \Admissibility[\alpha])$,
		where in $\transitionfrom[\alpha,N]{}{}$ the filter functions $\filtere[\alpha]{}{} \subset \filtere[N]{}{}$ and the neutral action filters $\filterag[N]{i}{}{}$ (for all $i \in \agents$) are used, $PP^\alpha \subset \envprotocols \times \agprotocols$ and $\extension[\alpha] \notin \efjp$.
	\item[$\efejpfb$] 
		An extension $\extension[\alpha] \in \efejpfb$ if{f}
			$\extension[\alpha]=(PP^\alpha, IS^\alpha, \transitionfrom[\alpha,B]{}{}, \Admissibility[\alpha])$,
		where in $\transitionfrom[\alpha,B]{}{}$ the filter functions $\filtere[\alpha]{}{} \subset \filtere[N]{}{}$ and the byzantine action filters $\filterag[B]{i}{}{}$ (for all $i \in \agents$) are used, $PP^\alpha \subset \envprotocols \times \agprotocols$ and $\extension[\alpha] \notin \efjpfb$.
	\item[$\others$]
		This class contains all remaining extension implementations including restrictions via arbitrary action filters $\filterag{i}{}{}$ (for $i \in \agents$).

		An extension $\extension[\alpha] \in \others$ if{f} it is not in any other class.
	\end{itemize}

	We list important subsets of these  implementation classes, which  we will treat as individual implementation classes in their own right (see listing below):\looseness=-1
	\begin{compactitem}
	\item $\jpdc \ce \{\extension[\alpha] \in \jp \mid S^\alpha \text{ is downward closed}\}$
	\item $\ejpdc \ce \{\extension[\alpha] \in \ejp \mid S^\alpha \text{ is downward closed}\}$
	\item $\efjpdc \ce \{\extension[\alpha] \in \efjp \mid S^\alpha \text{ is downward closed}\}$
	\item $\efejpdc \ce \{\extension[\alpha] \in \efejp \mid S^\alpha \text{ is downward closed}\}$
	\item $\othersdc \ce \{\extension[\alpha] \in \others \mid S^\alpha \text{ is downward closed}\}$
	\item $\efejpdcmono \ce \{\extension[\alpha] \in \efejpdc \mid (\forall i \in \agents)\ \filterag[\alpha]{i}{}{} \text{ and } \filtere[\alpha]{}{}\\
				\text{are monotonic for the domain } PD_\epsilon^{t-coh}, 2^{\gactions[1]}, \dots, 2^{\gactions[n]}\}$.
	\item $\othersdcmono \ce \{\extension[\alpha] \in \othersdc \mid (\forall i \in \agents)\ \filterag[\alpha]{i}{}{} \text{ and } \filtere[\alpha]{}{} \text{ are}\\
				\text{monotonic for the domain } PD_\epsilon^{t-coh}, 2^{\gactions[1]}, \dots, 2^{\gactions[n]}\}$.
	\end{compactitem}
\end{definition}

\begin{lemma}\label{lem:virtProt}
	An agent $i$ in a synchronous agents context executes its protocol only during synced rounds, i.e., for every $\chi \in \extension[S]$ and $r \in \system{\chi}$, $go(i) \in \betag[i]{r}{t}$ if $t$\textonehalf{} is a synced round.
\end{lemma}
\begin{proof}
	From Defs.~\ref{def:virtRound} and~\ref{def:filterSyn}, it immediately follows that in a synchronous agents context $go(i)$ events can only ever occur during a synced round. \qed
\end{proof}

\begin{lemma}\label{lem:virtPreAware}
	For a correct agent $i$, a $\transitionExt[\envprotocol{}, \joinprotocol{}^{S}]{S}{}$-transitional run $r$ (where $\joinprotocol{}^{S} \in \agprotocols[S]$), some timestamp $t' \ge 1$, agent $i$'s local history $r_i(t') = h_i = [\lambda_m, \dots, \lambda_1, \lambda_0]$ (given the global history $h = r(t') \in \globalstates$) and some round $(t-1)$\textonehalf{} ($t' \ge t \ge 1$), there exists some $a \in \actions[i]$ such that $a \in \lambda_{k_t}$ where $\lambda_{k_t} = \sigmaof{\betae[i]{r}{t-1} \sqcup \betaag{i}{r}{t-1}}$ if and only if $(t-1)$\textonehalf{} is a synced round.
\end{lemma}
\begin{proof}
	From left to right.
	From Lemma \ref{lem:virtProt}, we know that an agent can only execute its protocol during synced rounds.
	Therefore, since agent $i$ is assumed to be correct and $(t-1)$\textonehalf{} is a synced round, it follows that $\{go(i)\} = \betag[i]{r}{t-1}$ ($sleep(i)$ or $hibernate(i)$ would make the agent byzantine).
	By (\ref{eq:update_agent}) (the definition of the update function) and Def.~\ref{def:synch_ag_protocols} (the definition of the synchronous agents joint protocols, which dictates that at least $\clock$ has to be among the attempted actions, hence the empty set can never be issued) an action $a \in \actions[i]$ such that $a \in \lambda_{k_t}$ has to exist.
	
	From right to left.
	Suppose there exists  $a \in \actions[i]$ such that $a \in \lambda_{k_t}$.
	Since agent $i$ is assumed to be correct, by Lemma \ref{lem:virtProt} agents only execute their protocol during synced rounds and by the definition of the update function (\ref{eq:update_agent}) round $(t-1)$\textonehalf{} has to be a synced round. \qed
\end{proof}

\begin{lemma} \label{lem:synch_go_means_action}
	For any agent $i \in \agents$, any run $r \in \system{\chi}$, where $\chi \in \extension[S]$ and any timestamp $t \in \bbbt$, it holds that
		$\{go(i)\} = \betag[i]{r}{t}$ 
		if{f} $(\exists A \in \gtrueactions[i])\ A \in \betaag{i}{r}{t}$.
\end{lemma}
\begin{proof}
	This directly follows from Def.~\ref{def:synch_ag_protocols} of the synchronous agents joint protocol and the standard action filter function \ref{def:std_ac_filter}.
	As no synchronous agents protocol can prescribe the empty set, whenever an agent $i$ receives a $go(i)$ event during some round $t$\textonehalf{}, it will perform some action $a \in \actions[i]$, as by $t$-coherence of the environment protocol's event sets, there can always only be one system event present for any agent during one round.
	Similarly, if $(\exists A \in \gtrueactions[i])\ A \in \betaag{i}{r}{t}$ by definition of the byzantine action filter, $i$ must have gotten a $go(i)$. \qed
\end{proof}
\begin{corollary} \label{cor:synch_dist}
	Lemma \ref{lem:synch_gbyz_indist} does not hold for runs $r, r' \in \system{\chi}$ for $\chi \in \extension[S]$.
\end{corollary}

\begin{definition} \label{def:neut_filters}
	For   $i \in \agents$, global history $h \in \globalstates$, 
	we define the \textbf{neutral} event and action filters (the weakest filters) as
$		\filtere[N]{h}{X_\epsilon,\, X_\agents} \ce X_\epsilon $
and
	$	\filterag[N]{i}{X_\agents}{X_\epsilon} \ce X_i$.
	The transition template using only the neutral filters is denoted~$\transitionfrom[N,N]{}{}$ or $\transitionfrom[N]{}{}$.
\end{definition}

\subsection{Asynchronous Byzantine Agents} \label{sec:asynch_byz_ag}
\begin{definition}
    \label{def:asynch_byz_ag_ext}
    We denote by
$		\extension[B] \ce (\envprotocols\times\agprotocols, 2^{\globalstates(0)} \setminus \{\varnothing\}, \transitionfrom[B]{}{}, \system{})$
	the \textbf{asynchronous byzantine agents} extension.
\end{definition}

\begin{lemma} \label{lem:byz_ext_from_efjpfb}
	$\extension[B] \in \efjpfb$.
\end{lemma}
\begin{proof}
	Follows from Defs.~\ref{def:asynch_byz_ag_ext} and \ref{def:impl_class}.
\end{proof}

\subsection{Reliable Communication}
\label{sec:rc}
In the \textbf{reliable communication} extension agents can behave arbitrarily.
However the communication---the transmission of messages by the environment---is reliable for a particular set of (reliable) channels, i.e., a message that was sent through one of these (reliable) channels, is guaranteed to be delivered by the environment in finite time.
This also holds for the delivery of messages to and from byzantine agents.
Since a byzantine agent can always "choose" to ignore any messages it receives anyway, this does not restrict its byzantine power to exhibit arbitrary behaviour.
Formally, we define a set of (reliable) channels as $C \subseteq \agents^2$.\looseness=-1

The reliable communication property will be ensured by the admissibility condition $\EDelC$, which is a liveness property.
\begin{definition}[Eventual Message Delivery]
    \index{$\EDelC$}
    \label{def:EDelC}
	\begin{equation}\label{eq:EDelC}
\adjustbox{width=.88\textwidth}{$		\begin{aligned}
			\EDelC &= \left\{r \in \system{} \,\bigg|\, \bigg(\Big(\gsend{i}{j}{\mu}{id} \in r_\epsilon(t) \quad \vee \right. \\
			&(\exists A \in \{\tick\}\sqcup\gactions[i])\, \fakeof{i}{\mistakefor{\gsend{i}{j}{\mu}{id}}{A}}\in r_\epsilon(t) \Big) \quad \wedge \\
			&\left.(i,j) \in C\bigg) \longrightarrow (\exists t' \in \mathbb{N})\ \grecv{j}{i}{\mu}{id} \in r_\epsilon(t') \right\}
		\end{aligned}$}
	\end{equation}
\end{definition}

\begin{definition}
    \index{$\extension[RC_C]$}
    \label{def:reliable_channel_extension}
    We define by
$		\extension[RC_C] \ce \bigl(\envprotocols\times\agprotocols, 2^{\globalstates(0)} \setminus \{\varnothing\}, \transitionfrom[N]{}{}, \EDelC\bigr)
$
	the \textbf{reliable communication} extension.
\end{definition}

\subsection{Time-bounded Communication}
\label{sec:tb_communication}
\newcommand{\net}{{\Delta}}
\newcommand{\netmapsto}{$\mapsto^{\net}$}
\newcommand{\timebound}[4]{%
    \delta_{#1\mapsto#2}\ifstrempty{#3}{}{\left({#3},{#4}\right)}}
\newcommand{\TB}{TB}

We say that communication is time-bounded if for every channel and for every message there is an upper-bound (possibly infinite) on the transmission time.
Since the transmission is not reliable a priori, the \textbf{time-bounded communication} extension only specifies the time window during which the delivery of a message can occur.
In order to gain flexibility, bounds can be changed depending on the sending time and depending on the message too---for instance a byte of data and picture will not have the same time bound.
We encode these bounds in an upper-bound structure defined as follows:
\begin{definition}
	\label{def:timebound}
    \index{$\timebound{i}{j}{}{}$}
    \index{$\net$}
    For the first infinite ordinal number $\omega$,  agents $(i,j)\in\agents^2$, and the channel $i\mapsto j$, we define the message transmission upper-bound for the channel $i\mapsto j$ as follows
 $       \timebound{i}{j}{}{} \colon \msgs\times\mathbb{N} \to  \mathbb{N}\cup\{\omega\}.
 $
    We define an upper bound structure as
$		\net \ce \bigcup_{(i,j)\in\agents^2}{\{\timebound{i}{j}{}{}{}}\}
$.
\end{definition}

Since, as we soon show, the time-bounded safety property is downward closed, we implement it by restriction of the set of environment protocols.\looseness=-1
\begin{definition}
	For an upper-bound structure $\net$, we define the set of \textbf{time-bounded communication environment protocols} as
	\label{def:tbc_env_prot}
    \begin{multline} \label{eq:tbc_env_prot}
			\envprotocols[TC_\net] \ce \{P_\epsilon \in \envprotocols \mid (\forall t \in \mathbb{N})(\forall X_\epsilon \in P_\epsilon(t))\\
				\grecv{j}{i}{\mu}{id(i,j,\mu,k,t')} \in X_\epsilon \ \rightarrow \ t' + \timebound{i}{j}{\mu}{t'} \ge t\}.
    \end{multline}
\end{definition}

\begin{definition} \label{def:timebounded_com_extension}
    \index{$\extension[TC_\net]$}
    For an upper-bound structure $\net$
	\begin{equation*} 
		\extension[TC_\net] \ce (\envprotocols[TC_\net]\times\agprotocols, 2^{\globalstates(0)} \setminus \{\varnothing\}, \transitionfrom[N]{}{}, \System)
	\end{equation*}
	denotes the \textbf{time-bounded communication} extension.
\end{definition}

\begin{lemma} \label{lem:saf_tc_downward_closed}
	$S^{TC_\net}$ is downward closed.
\end{lemma}
\begin{proof}
	Suppose that by contradiction $S^{TC_\net}$ is not downward closed.
	This implies $X' \notin S^{TC_\net}(h)$ for some $h \in \globalstates$,  $X \in S^{TC_\net}$, and  $X' \subseteq X$. 
	It immediately follows that $X' \subset X$.
	Since in $\transitionfrom[N]{}{}$ the neutral (event and action) filters are used we further get that there are $P_\epsilon \in \envprotocols[TC_\net]$ and $P \in \agprotocols$,
$X_\epsilon \in P_\epsilon(|h|)$,  $X_i \in P_i(h_i)$ for all $i \in \agents$, and 
$			X = X_\epsilon \sqcup X_1 \sqcup \dots \sqcup X_n$.
%
	Since the set of joint protocols is unrestricted there exists some joint protocol~$P'$ ensuring that together with $P_\epsilon$,
$X_\epsilon \sqcup X'_1 \sqcup \dots \sqcup X'_n \in S^{TC_\net}(h)$ for all 
		$X'_i \subseteq X_i$, $i \in \agents$.
	Therefore, we conclude that the violation of $X'$ has to be caused by some $X'_\epsilon \subset X_\epsilon = X \sqcup \gevents$.

	From $X \in S^{TC_\net}(h)$ we conclude that
	\begin{equation} \label{eq:tbc_timebound_holds_for_X}
		\grecv{j}{i}{\mu}{id(i,j,\mu,k,t')} \in X_\epsilon \ \rightarrow \ t' + \timebound{i}{j}{\mu}{t'} \ge |h|.
	\end{equation}

	By semantics of "$\rightarrow$" and since $X'_\epsilon \subseteq X_\epsilon$ we get
\begin{equation} \label{eq:tbc_X'_implies_X}	
\adjustbox{width=.88\textwidth}{	
$			\bigl(\grecv{j}{i}{\mu}{id(i,j,\mu,k,t')} \in X'_\epsilon\bigr) \rightarrow \bigl(\grecv{j}{i}{\mu}{id(i,j,\mu,k,t')} \in X_\epsilon\bigr).$
}	\end{equation}
	Using (\ref{eq:tbc_X'_implies_X}) in (\ref{eq:tbc_timebound_holds_for_X}) we get
	\begin{equation} \label{eq:tbc_X'_implies_X_bound_holds}
\adjustbox{width=.88\textwidth}{	
$		\begin{aligned}
			&\bigl(\grecv{j}{i}{\mu}{id(i,j,\mu,k,t')} \in X'_\epsilon\bigr) \rightarrow \bigl(\grecv{j}{i}{\mu}{id(i,j,\mu,k,t')} \in X_\epsilon\bigr) \rightarrow \\
			&\bigl(t' + \timebound{i}{j}{\mu}{t'} \ge |h|\bigr).
		\end{aligned}$
}
	\end{equation}
	Finally from (\ref{eq:tbc_X'_implies_X_bound_holds}) by transitivity of "$\rightarrow$" we get that
	\begin{equation}
		\bigl(\grecv{j}{i}{\mu}{id(i,j,\mu,k,t')} \in X'_\epsilon \bigr) \rightarrow \bigl( t' + \timebound{i}{j}{\mu}{t'} \ge |h|\bigr).
	\end{equation}
	Hence, we conclude that $X' \in S^{TC_\net}(h)$ and we are done.
\end{proof}

\begin{corollary} \label{cor:tbc_impl_ejpdc}
	$\extension[TC_\net] \in \ejpdc$.
\end{corollary}
\begin{proof}
	Follows from Def.~\ref{def:timebounded_com_extension} and Lemma \ref{lem:saf_tc_downward_closed}.
\end{proof}

\subsection{Synchronous Communication}
\label{sec:synchronous_communication}
The \textbf{Synchronous Communication} extension guarantees for a set of synchronous communication channels $C \subseteq \agents^2$ that whenever a message is correctly received, it has been sent during the same round.
This means that it is a special case of the time-bounded communication extension.
\begin{definition}[Synchronous Communication Environment Protocols]
	We define the synchronous message delay as
	\begin{equation*}
		\delta^{\mathit{SC}_C}_{i \mapsto j}(\mu,t) \ce
		\begin{cases}
			0 &\text{if } (i,j) \in C \\
			\omega &\text{otherwise}
		\end{cases}
	\end{equation*}
	We define the synchronous communication upper bound structure $\net^{\mathit{SC}_C}$ as
$		\net^{\mathit{SC}_C} \ce \bigcup_{(i,j)\in\agents^2}\{\delta^{\mathit{SC}_C}_{i \mapsto j}\}$.
    \begin{equation} \label{eq:sc_env_prot}
		\envprotocols[\mathit{SC}_C] \ce \envprotocols[TC_{\net^{\mathit{SC}_C}}]
    \end{equation}
\end{definition}    

\begin{definition}
    \index{$\extension{SC_C}$}
    \label{def:sc_extension}
    We denote by
$		\extension[\mathit{SC}_C] \ce \left(\envprotocols[\mathit{SC}_C]\times\agprotocols, 2^{\globalstates(0)} \setminus \{\varnothing\}, \transitionfrom[N]{}{}, \System\right)
$ 
	the \textbf{synchronous communication} extension.
\end{definition}

\begin{lemma} \label{saf_sc_downward_closed}
	$S^{\mathit{SC}_C}$ is downward closed.
\end{lemma}
\begin{proof}
	Follows from Lemma \ref{lem:saf_tc_downward_closed}, as the synchronous communication extension is just an instance of the time-bounded communication extension (\ref{eq:sc_env_prot}).
\end{proof}

\begin{corollary} \label{cor:SC_impl_class_ejpdc}
	$\extension[\mathit{SC}_C] \in \ejpdc$.
\end{corollary}
\begin{proof}
	Follows from Def.~\ref{def:sc_extension} and Lemma \ref{saf_sc_downward_closed}.
\end{proof}

\subsection{Multicast Communication}
\label{sec:multicast_communication}
In the \textbf{multicast communication} paradigm, each agent has several multicast channels at its disposal and is restricted to sending messages using these particular channels.
In this section, we provide a software based multicast, meaning that only correct agents have to adhere to this behavior (further along we provide a hardware based multicast as well, where also byzantine agents are forced to exhibit this multicast behavior).

First, we define a \textbf{multicast communication problem}.
For each $i\in \agents$ we define a collection $Mc_i$ of groups of agents it can send messages to. 
\begin{definition}
    For each $i\in \agents$ the set of available multicast channels is
    $    Mc_i \subseteq 2^{\agents}\setminus \{\varnothing\}$.
    The \textbf{multicast communication problem} is the tuple of these collections of communication channels
       $ Ch = (Mc_1,\dots,Mc_n)$.
\end{definition}

\newcommand{\recipient}[2]{Rec_{#1}(#2)}
We denote the set of recipients for the copy $\mu_k$ of a message $\mu$   that has been sent according to some set $X \subseteq \actions$ by 
$	\recipient{X}{\mu_k} = \{j \mid \send{j}{\mu_k} \in X \}$.

Since we implement a software based multicast (and since we want our extensions to be modular) we  use a restriction of the joint protocol to do so.
\begin{definition}
    For a multicast communication problem $Ch$, we define the set of \textbf{multicast joint protocols} as
    \index{${\agprotocols[MC_{Ch}]}$}
    \begin{equation}
        \label{eq:mc_ag_protocols}
\adjustbox{width=.88\textwidth}{	
	$	\begin{aligned}
        \agprotocols[\mathit{MC}_{Ch}] = \{
           (\agprotocol{1}{},\dots,\agprotocol{n}{})\in \agprotocols \mid  &(\forall i \in \agents)(\forall h_i \in \localstates{i})(\forall X \in \agprotocol{i}{h_i})(\forall \mu \in \msgs)(\forall k \in \mathbb{N}) \\
																		   &\recipient{X}{\mu_k} \ne \varnothing \ \rightarrow \ \recipient{X}{\mu_k}\in Mc_i \}.
		\end{aligned}$}
    \end{equation}
\end{definition}

\begin{definition}
    \index{$\extension[MC_{Ch}]$}
    \label{def:multicast_extension}
    For a multicast communication problem $Ch$, we set 
$		\extension[\mathit{MC}_{Ch}] \ce \bigl(\envprotocols\times\agprotocols[\mathit{MC}_{Ch}], 2^{\globalstates(0)} \setminus \{\varnothing\}, \transitionfrom[N,B]{}{}, \System\bigr)$
to be the \textbf{multicast communication} extension
	where in $\transitionfrom[N,B]{}{}$ the neutral event and the byzantine action filters are used (for all $i \in \agents$).
\end{definition}

An important special case of the multicast communication problem is the broadcast communication problem, where each agent must broadcast each message to all the agents:
\begin{definition}
    \index{$\extension[BC]$}
    \label{def:broadcast_extension}
The \textbf{broadcast communication extension} $\extension[BC]$ is a multicast communication extension $\extension[\mathit{MC}_{\mathit{BCh}}] $ for 
	\begin{equation} \label{eq:mc_bc}
		\mathit{BCh}= (\underbrace{\{\agents\},\dots,\{\agents\}}_n)
	\end{equation}
$		\extension[BC]= (\envprotocols\times\agprotocols[\mathit{MC}_{\mathit{BCh}}], 2^{\globalstates(0)} \setminus \{\varnothing\}, \transitionfrom[N,B]{}{}, \System)$,
	where in $\transitionfrom[N,B]{}{}$ the neutral event filter and the byzantine action filters (for all $i \in \agents$) are used.
\end{definition}

\begin{corollary} \label{cor:mc_jpfb}
	$\extension[\mathit{MC}_{Ch}] \in \jpfb$.
\end{corollary}
\begin{proof}
	Follows from Def.~\ref{def:multicast_extension}.
\end{proof}

\subsection{Lock-step Synchronous Agents}

\begin{table}[t]
	\caption{filter dependencies in the lock-step synchronous agents extension}
	\label{tab:filter_dep}
	\begin{tabular}{| l | l | l |}
		\hline
		\textbf{filter} & \textbf{dependency} & \textbf{removal} \\ \hline 
		$\filtere[S]{}{}$ & $go(i), sleep(i), hibernate(i)$ & $go(i)$ \\ \hline
		$\filtere[B]{}{}$ & $go(i), \gsend{i}{j}{\mu}{id}, \fakeof{i}{\mistakefor{\gsend{i}{j}{\mu}{id}}{A}}$ & $\grecv{j}{i}{\mu}{id}$ \\ \hline
	\end{tabular}
\end{table}

Table \ref{tab:filter_dep} reveals that $\filtere[B]{}{}$ depends on $go(i)$ events, which $\filtere[S]{}{}$ removes.
Thus, we have a dependence relation from $\filtere[B]{}{}$ to $\filtere[S]{}{}$.
$\filtere[B]{}{}$ removes only correct receive events $\grecv{j}{i}{\mu}{id}$.
$\filtere[S]{}{}$ is independent of such events, hence, there is no dependence relation from $\filtere[S]{}{}$ to $\filtere[B]{}{}$.\looseness=-1

\begin{figure}
	\begin{center}
		\begin{tikzpicture}[
				> = stealth, 
				shorten > = 1pt, 
				auto,
				node distance = 3cm, 
				semithick 
			]

			\tikzstyle{every state}=[
				draw = black,
				thick,
				fill = white,
				minimum size = 4mm
			]

			\node[state] (b) {$\filtere[B]{}{}$};
			\node[state] (s) [right of=b] {$\filtere[S]{}{}$};

			\path[->] (b) edge node {} (s);
		\end{tikzpicture}
	\end{center}
	\caption{Dependence graph for $\filtere[B]{}{}$ and $\filtere[S]{}{}$}
	\label{fig:dep_B_S}
\end{figure}

Figure \ref{fig:dep_B_S} shows the final dependence graph.
Since there is no circular dependence, we can directly use the composition order given by the graph.
This gives us \looseness=-1
$	\extension[B \circ S] = (\envprotocols \times \agprotocols[S], 2^{\globalstates(0)} \setminus \{\varnothing\}, \transitionfrom[B \circ S, B]{}{}, \system{})$,
where in $\transitionfrom[B \circ S, B]{}{}$ the event filter is $\filtere[B \circ S]{}{}$ and the action filters result in $\filterag[B]{i}{}{}$ for all $i \in \agents$ (by idempotence of the byzantine action filter function).
Following the rest of the extension combination guide finally leads to
\begin{equation}
\adjustbox{width=.88\textwidth}{
$	\extension[B \circ S \circ BC \circ \mathit{SC}_{\agents^2} \circ RC_{\agents^2}] = \bigl(\envprotocols[\mathit{SC}_{\agents^2}]\times(\agprotocols[\mathit{MC}_{\mathit{BCh}}] \cap \agprotocols[S]), 2^{\globalstates(0)} \setminus \{\varnothing\},\transitionfrom[B \circ S, B]{}{}, \EDel_{\agents^2} \bigr).
$}
\end{equation}

\begin{lemma} \label{lem:lss_comb_saf}
	The extensions $\extension[B]$, $\extension[S]$, $\extension[\mathit{SC}_{\agents^2}]$, $\extension[RC_{\agents^2}]$, and $\extension[BC]$ are compatible (w.r.t. the composition $B \circ S \circ BC \circ \mathit{SC}_{\agents^2} \circ RC_{\agents^2}$).
\end{lemma}
\begin{proof}
	The only condition from Def.~\ref{def:ext_compat} that does not trivially follow from the definition of the extensions in question is whether there exists an agent context $\chi$, such that
	$\chi \in \extension[B \circ S \circ BC \circ \mathit{SC}_{\agents^2} \circ RC_{\agents^2}]$.
	Such a $\chi$ however can be easily constructed.
	Let $\chi = \left((\envprotocol{}',\globalinitialstates,\transitionfrom[B \circ S, B]{}{}, \EDel_{\agents^2}), \joinprotocol{}'\right)$, where $\envprotocol{}'$ only produces the set containing the empty set and $\joinprotocol{}'$ for every agent produces the set containing the set that only contains the action $\clock$, i.e., 
$		
P'_\epsilon(t) = \{\varnothing\}$ for all $t \in \mathbb{N}$ and $P'(h) = (\{\{\clock\}\},\,\dots,\,\{\{\clock\}\})$ for all $ h \in \globalstates$.
	Note that this agent context is part of the extension $\extension[B \circ S \circ BC \circ \mathit{SC}_{\agents^2} \circ RC_{\agents^2}]$, as \looseness=-1
$		\envprotocol{}' \in \envprotocols[\mathit{SC}_{\agents^2}]$  
and $\joinprotocol{}' \in (\agprotocols[\mathit{MC}_{\mathit{BCh}}] \cap \agprotocols[S])$.
\end{proof}

\begin{lemma} \label{lem:lss_sat_all_saf}
	The extension $\extension[B \circ S \circ BC \circ \mathit{SC}_{\agents^2} \circ RC_{\agents^2}]$ satisfies all safety properties of its constituent extensions.
\end{lemma}
\begin{proof}
	Follows from Table \ref{tab:composability_matrix}.
\end{proof}

Finally, after having proved that the resulting extension $\extension[B \circ S \circ BC \circ \mathit{SC}_{\agents^2} \circ RC_{\agents^2}]$ satisfies all desired properties, we can define it as $\extension[\mathit{LSS}]$.
\begin{definition} \label{def:lss_ext}
    \index{$\extension[\mathit{LSS}]$}
    We define the \textbf{lock-step synchronous agents} extension to be $\extension[\mathit{LSS}]=\bigl(\envprotocols[\mathit{SC}_{\agents^2}]\times(\agprotocols[\mathit{MC}_{\mathit{BCh}}] \cap \agprotocols[S]), 2^{\globalstates(0)} \setminus \{\varnothing\},\transitionfrom[B \circ S, B]{}{}, \EDel_{\agents^2} \bigr)$.
\end{definition}

We will now add a few lemmas about properties, which the lock-step synchronous agents extension inherits from the synchronous agents extension.

\begin{lemma}\label{lem:lss_virtProt}
	An agent $i$ in a lock-step synchronous agents context executes its protocol only during synced rounds, i.e., $go(i) \in \betag[i]{r}{t}$ if{f} $t.5$ is a synced round.\looseness=-1
\end{lemma}
\begin{proof}
 Lemma~\ref{lem:virtProt} for the synchronous agents extension describes a property of $S^S$ that by Lemma \ref{lem:lss_sat_all_saf}, $\extension[\mathit{LSS}]$ satisfies.
\end{proof}

\begin{lemma}\label{lem:lss_virtPreAware}
	For a correct agent $i$,
	a $\transitionExt[\envprotocol{}^{\mathit{SC}_{\agents^2}}, \joinprotocol{}^{\mathit{SMC}_{\mathit{BCh}}}]{B \circ S, B}{}$-transitional run $r$ (where $\envprotocol{}^{\mathit{SC}_{\agents^2}} \in \envprotocols[\mathit{SC}_{\agents^2}]$ and $\joinprotocol{}^{\mathit{SMC}_{\mathit{BCh}}} \in \agprotocols[S] \cap \agprotocols[\mathit{MC}_{\mathit{BCh}}]$),
	some timestamp $t' \ge 1$,
	agent~$i$'s local history $r_i(t') = h_i = [\lambda_m, \dots, \lambda_1, \lambda_0]$ (given the global history $h = r(t') \in \globalstates$) and
	some round $(t-1)$\textonehalf{} ($t' \ge t \ge 1$),
	there exists some $a \in \actions[i]$ such that $a \in \lambda_{k_t}$ where $\lambda_{k_t} = \sigmaof{\betae[i]{r}{t-1} \sqcup \betaag{i}{r}{t-1}}$ if and only if $(t-1)$\textonehalf{} is a synced round.
\end{lemma}
\begin{proof}
	This again follows from Lemma \ref{lem:virtPreAware} for the synchronous agents extension, as the statement of this lemma is a safety property of $\extension[S]$ and by Lemma~\ref{lem:lss_sat_all_saf}, $\extension[\mathit{LSS}]$ satisfies $S^S$.
\end{proof}

\begin{lemma} \label{lem:lock_step_synch_go_means_action}
	For any agent $i \in \agents$, any run $r \in \system{\chi}$, where $\chi \in \extension[\mathit{LSS}]$ and any timestamp $t \in \mathbb{N}$ it holds that
$		go(i) \in \betag[i]{r}{t} \Longleftrightarrow (\exists A \in \gtrueactions[i])\, A \in \betaag{i}{r}{t}.$
\end{lemma}
\begin{proof}
	Analogous to the proof of Lemma \ref{lem:synch_go_means_action} for synchronous agents.\looseness=-1
\end{proof}

\begin{lemma}[Lock-step Synchronous Brain-in-the-Vat Lemma]
\label{lem:lock_step_synch_compose-fake-freeze}
	Let $\agents=\llbracket 1;n\rrbracket$ be a set of agents with joint protocol $\joinprotocol{}{}^{\mathit{SMC}_{\mathit{BCh}}} = (\agprotocol{1}{},\dots,\agprotocol{n}{}) \in (\agprotocols[MC_{\mathit{BCh}}] \cap \agprotocols[S])$,
	let $\envprotocol{}^{\mathit{SC}_{\agents^2}} \in \envprotocols[\mathit{SC}_{\agents^2}]$ be the protocol of the environment,
    for $\chi \in \extension[\mathit{LSS}]$, where $\chi = ((\envprotocol{}^{\mathit{SC}_{\agents^2}}, \globalstates(0), \transitionfrom[B \circ S, B]{}{}, \EDelC), \joinprotocol{}{}^{\mathit{SMC}_{\mathit{BCh}}})$, let $r \in \system{\chi}$,
	let $i \in \agents$ be an agent, let $t>0$ be a timestamp and let $\adj=[B_{t-1}; \dots ;B_0]$ be an adjustment of extent $t-1$ satisfying
$		B_m = (\rho^m_1,\, \dots,\, \rho^m_n) $
	for all $0 \leq m \leq t-1$ with
    $\rho^{m}_i=\improvednewfakerule{i}{m}{}{}$ and 
    for all $j\ne i$ 
    $\rho^{m}_j \in\{\newfreezerule{},\newffreezerule{j}{}\}$.
If the protocol $\envprotocol{}^{\mathit{SC}_{\agents^2}}$ makes
\begin{compactitem} 
\item agent $i$ gullible,
\item every agent $j\ne i$ delayable and fallible if $\rho^m_j = \newffreezerule{j}{}$ for some $m$, 
\item all remaining agents delayable, 
\end{compactitem} 
then each run $r'\in R\bigl({\tauprotocol{B \circ S, B}{\envprotocol{}^{\mathit{SC}_{\agents^2}}}{\joinprotocol{}^{\mathit{SMC}_{\mathit{BCh}}}}},{r},{\adj}\bigr)$ satisfies the following properties:
    \begin{compactenum}
		\item $r' \in \system{\chi}$.
        \item\label{lem:synch_compose-fake-freeze:i-same} $(\forall m \le t)\  \run[']{i}{m}=\run{i}{m}$;
        \item\label{lem:synch_compose-fake-freeze:j-frozen} $\left(\forall m \le t)(\forall j \neq i \right)\  \run[']{j}{m}=\run[']{j}{0}$;
		\item\label{lem:synch_compose-fake-freeze:i-bad} $(i,1) \in \failed{r'}{1}$ and thus $(i,m) \in \failed{r'}{m'}$ for all $m'\geq m>0$;
		\item\label{lem:synch_compose-fake-freeze:only-i-failed} $\agentsof{\failed{\run[']{}{t}}{}} = \{i\} \cup \{j \ne i \mid (\exists m \leq t-1)\  \rho^m_j = \newffreezerule{j}{}\}$;
		\item\label{lem:synch_compose-fake-freeze:j-no-events} $\left(\forall m < t)\  (\forall j \neq i \right)\  \betae[j]{r'}{m}\subseteq \{\failof{j}\}$. 
		More precisely, $\betae[j]{r'}{m}=\varnothing$ if{f} $\rho^m_j = \newfreezerule{}$ and  $\betae[j]{r'}{m} = \{\failof{j}\} $ if{f} $\rho^m_j = \newffreezerule{j}{}$;
        \item\label{lem:synch_compose-fake-freeze:i-no-correct-events} $(\forall m < t)\  \betae[i]{r'}{m}\setminus\betaf[i]{r'}{m} =\varnothing$;
        \item\label{lem:synch_compose-fake-freeze:all-no-actions} $(\forall m < t) (\forall j \in \agents)\ \betaag{j}{r'}{m}=\varnothing$.
    \end{compactenum}
\end{lemma}
\begin{proof}
	The proof is (similar to Lemma \ref{lem:synch_compose-fake-freeze}) analogous to the original Brain-in-the-Vat Lemma~\cite{KPSF19:FroCos}
	, since by Def.~\ref{def:bitv_interv} of the interventions $\newfreezerule{}$, $\newffreezerule{i}{}$, and $\improvednewfakerule{i}{t}{}$ it holds that
	\begin{equation} \label{eq:lss_adj_no_goes_or_corr_ev}
		\begin{gathered}
			\bigl(\forall r'\in R\bigl({\tauprotocol{B \circ S, B}{\envprotocol{}^{\mathit{SC}_{\agents^2}}}{\joinprotocol{}^{\mathit{SMC}_{\mathit{BCh}}}}},{r},{\adj}\bigr)\bigr)(\forall j \in \agents)(\forall m \in \mathbb{N} \text{ s.t. } 0 \le m < t) \\
			go(j) \notin \betae[j]{r'}{m} \ \wedge \ \betaout[j]{r'}{m} = \varnothing.
		\end{gathered}
	\end{equation}
	By Def.~\ref{def:lss_ext} of the lock-step synchronous agents extension both its set of environment protocols and its admissibility condition from Def.~\ref{def:EDelC} only restrict runs (respectively environment protocols) w.r.t. correct receive events (see \eqref{eq:sc_env_prot} and \eqref{eq:tbc_env_prot}).
	By \eqref{eq:lss_adj_no_goes_or_corr_ev} however correct events do not event occur in any such runs $r'$.
	Furthermore the synchronous agents event filter function by Def.~\ref{def:filterSyn} only additionally removes $go$ events, which by \eqref{eq:lss_adj_no_goes_or_corr_ev} also are irrelevant for such runs $r'$.
	Additionally \eqref{eq:lss_adj_no_goes_or_corr_ev} makes the set of joint protocols superfluous for this lemma, hence the proof from~\cite{KPSF19:FroCos} applies for the lock-step synchronous agents extension as well.
\end{proof}

Here are some new properties unique to the lock-step synchronous extension.

\begin{lemma} \label{lem:lss_ag_send_inst_recv}
	Whenever a correct agent $i \in \agents$ in an agent context $\chi \in \extension[\mathit{LSS}]$ sends a message $\mu$ in round $t$, it sends $\mu$ to all agents and $\mu$ is received by all agents in the same round $t$.
\end{lemma}
\begin{proof}
	When a correct agent $i$ sends a message, this is done by executing its protocol (as a fake send initiated by the environment protocol would immediately make this agent faulty).
	From the definition of the joint protocol (Def.~\ref{def:synch_ag_protocols}, (\ref{eq:mc_ag_protocols}) with (\ref{eq:mc_bc})), an agent can only send a message to all agents or no one.
	From the admissibility condition $\EDel_{\agents^2}$ (\ref{eq:EDelC}) and the synchronous communication environment protocol (\ref{eq:sc_env_prot}), it follows that a sent message has to be delivered to the receiving agent during the same round $t$ it was sent.
	Suppose by contradiction that a message, sent in round $t$, is not received by some agent in round~$t$.
	By (\ref{eq:EDelC}), it follows that this message has to be correctly received at some later point in time $t' > t$.
	However by (\ref{eq:sc_env_prot}), a correct receive event can only happen during the same round of its corresponding send event, thus leading to a contradiction.\looseness=-1
\end{proof}


\begin{theorem}\label{thm:lss_virtAware}
	A correct agent $i$ with local history $h_i$ in a lock-step synchronous agents context can infer from $h_i$ the number of synced rounds that have passed.
	Formally, for an agent context $\chi \in \extension[\mathit{LSS}]$, a $\chi$-based interpreted system $\I = (\system{\chi}, \pi)$, a run $r \in \system{\chi}$ and timestamp $t \in \mathbb{N}$,
	$	(\intsys,r,t) \models H_{i}{\nsr{\NSR{r(t)}}}$.
\end{theorem}
\begin{proof}
	Analogous to the proof of Theorem~\ref{lem:virtAware} from the synchronous agents extension as $\agprotocols[\mathit{MC}_{\mathit{BCh}}] \cap \agprotocols[S] \subseteq \agprotocols[S]$.
\end{proof}

\begin{theorem} \label{thm:lock_step_synch_gl_int_K_faulty_appendix}
	(Copy of Theorem \ref{thm:lock_step_synch_gl_int_K_faulty})\\
	There exists an agent context $\chi \in \extension[\mathit{LSS}]$, where $\chi = \bigl((\envprotocol{}^{\mathit{SC}_{\agents^2}}, \allowbreak \globalinitialstates, \allowbreak \transitionfrom[B \circ S]{}{}, \allowbreak \EDel_{\agents^2}),\widetilde{\joinprotocol{}}^{\mathit{SMC}_{\mathit{BCh}}}\bigr)$,
	and a run $r \in \system{\chi}$, such that
	for agents $i,j \in \agents$, where $i \ne j$,
	some timestamp $t \in \mathbb{N}$,
	and a $\chi$-based interpreted system $\intsys = (\system{\chi}, \interpretation{}{})$
	\begin{equation*} \label{equ:lock_step_synch_gl_int_K_faulty}
		\kstruct{}{r}{t} \models H_{i}{\faulty{j}}.
	\end{equation*}
\end{theorem}
\begin{proof}
	Suppose the joint protocol is such that for all global histories $h \in \globalstates$
	\begin{equation} \label{eq:lss_sp_proto}
		\begin{gathered}
			\widetilde{\joinprotocol{}}^{\mathit{SMC}_{\mathit{BCh}}}(h) = \{(S_1,\, \dots,\, S_n) \mid \\
			(\forall i \in \agents)(\forall D \in S_i)(\exists \mu \in \msgs)\ \{\send{j}{\mu} \mid (\forall j \in \agents)\} \cup \{\clock\} \subseteq D\}.
		\end{gathered}
	\end{equation}
	meaning that every agent has to perform at least one broadcast in case it gets the opportunity to act.
	By Lemma \ref{lem:lss_virtProt} (agents execute their protocols only during synced rounds), Lemma \ref{lem:lss_ag_send_inst_recv} (whenever a message is sent by a correct agent, all agents receive it during the same round) and (\ref{eq:lss_sp_proto}), it follows that every agent receives at least one message from every correct agent during a synced round.
	\textit{Thus in all states, where $i$ is correct, it received a message from itself, but not from some agent $j$, $j$ has to be faulty.}
\end{proof}

\end{document}